\documentclass[11pt]{article}
\usepackage{fullpage}
\usepackage{epsfig}
\usepackage{graphics}
\usepackage{latexsym}
\usepackage{amsmath}
\usepackage{amsfonts}
\usepackage{amssymb}
\usepackage{mathrsfs}
\usepackage{pifont}
\usepackage{amsthm}
\usepackage[usenames,dvipsnames]{color}
\usepackage{xspace}
\usepackage[ruled]{algorithm2e}
\usepackage{stmaryrd}
\usepackage{fontenc}
\usepackage{epstopdf}

\numberwithin{equation}{section}

\usepackage{amsthm}     
\theoremstyle{plain}     
\newtheorem{theorem}{Theorem}

\newtheorem{corollary}{Corollary}

\newtheorem{lemma}{Lemma}
\newtheorem{proposition}{Proposition}
\theoremstyle{definition} 
\newtheorem{definition}{Definition}

\newtheorem{example}{Example}

\newtheorem{observation}{Observation}
\theoremstyle{remark} 

\newtheorem{claim}{Claim}

\makeatletter
\@addtoreset{equation}{section}
\def\section{\@startsection {section}{1}{\z@}{-3.5ex plus -1ex minus
 -.2ex}{2.3ex plus .2ex}{\large\bf}}
\makeatother

\def\bfm#1{\mbox{\boldmath$#1$}}

\def\0{\bfm 0}

\newcommand{\sm}{\hspace{-2.5pt}\setminus\hspace{-2.5pt}}

\DeclareMathAlphabet{\mathpzc}{OT1}{pzc}{m}{it}

\title{A Network Game of Dynamic Traffic \thanks{This work was supported in part by the National Natural Science Foundation of China (NSFC),
  under grant numbers 11601022,  11531014 and 11471326.} }
\author{Zhigang Cao\footnote{Academy of Mathematics and Systems Science, Chinese Academy of Sciences, Beijing, {100190}, China, Email: zhigangcao@amss.ac.cn}, \  \  Bo Chen\footnote{Warwick Business School, University of Warwick, Coventry, {CV4 7AL,} United Kingdom, Email: Bo.Chen@wbs.ac.uk}, \ \  Xujin Chen\footnote{Academy of Mathematics and Systems Science, Chinese Academy of Sciences, Beijing,  {100190}, China, Email: xchen@amss.ac.cn},\  \  Changjun Wang\footnote{Beijing Institute for Scientific and Engineering Computing, Beijing University of Technology, Beijing, {100124}, China, Email: wcj@bjut.edu.cn}}
\date{\today}

\begin{document}

\maketitle

\bibliographystyle{plain}

\newcounter{my}
\newenvironment{mylabel}
{
\begin{list}{(\roman{my})}{
\setlength{\parsep}{-0mm}
\setlength{\labelwidth}{8mm}
\usecounter{my}}
}{\end{list}}

\newcounter{my2}
\newenvironment{mylabel2}
{
\begin{list}{(\alph{my2})}{
\setlength{\parsep}{-1mm} \setlength{\labelwidth}{8mm}
\setlength{\leftmargin}{6mm}
\usecounter{my2}}
}{\end{list}}

\newcounter{my3}
\newenvironment{mylabel3}
{
\begin{list}{(\alph{my3})}{
\setlength{\parsep}{-0mm}
\setlength{\labelwidth}{8mm}
\setlength{\leftmargin}{14mm}
\usecounter{my3}}
}{\end{list}}

\newcounter{my4}
\newenvironment{mylabel4}
{
\begin{list}{(\alph{my4})}{
\setlength{\parsep}{-0mm}
\setlength{\labelwidth}{2mm}
\setlength{\leftmargin}{3mm}
\usecounter{my4}}
}{\end{list}}

\newcounter{my5}
\newenvironment{mylabel5}
{
\begin{list}{(\alph{my5})}{
\setlength{\parsep}{-0mm}
\setlength{\labelwidth}{4mm}
\setlength{\leftmargin}{4mm}
\usecounter{my5}}
}{\end{list}}

\newcounter{my6}
\newenvironment{mylabel6}
{
\begin{list}{(\roman{my6})}{
\setlength{\parsep}{-0mm}
\setlength{\labelwidth}{8mm}
\setlength{\leftmargin}{8mm}
\usecounter{my6}}
}{\end{list}}


\begin{abstract}

We study a network congestion game of discrete-time dynamic traffic of atomic agents with a single origin-destination pair. Any agent freely makes a dynamic decision at each vertex (e.g., road crossing) and traffic is regulated with given priorities on edges (e.g., road segments). We first constructively prove that there always exists a subgame perfect equilibrium (SPE) in this game. We then study the relationship between this model and a simplified model, in which agents select and fix an origin-destination path simultaneously. We show that the set of Nash equilibrium (NE) flows of the simplified model is a proper subset of the set of SPE flows of our main model. We prove that each NE is also a strong NE and hence weakly Pareto optimal. We establish several other nice properties of NE flows, including global First-In-First-Out. Then for two classes of networks, including series-parallel ones, we show that the queue lengths at equilibrium are bounded at any given instance, which means the price of anarchy of any given game instance is bounded, provided that the inflow size never exceeds the network capacity.
\end{abstract}

\section{Introduction}

Selfish routing is one of the fundamental models in the {study of network traffic systems} {\cite{r07,rt02,w52}.} Most of its literature
assumes {\em essentially static} traffic flows, {failing} to capture many dynamic characteristics of traffics, 
despite the usual justification 
that static flows are good approximations of steady
states of the traffic system. Recently, models of \emph{dynamic} flows {have been drawing much  attention in various research areas of selfish routing} 
\cite{cc15,hm11,ks11,sst16,whk14}.

\subsection{Model}

Let us start with an informal and intuitive description of our model, with a formal definition 
{postponed to} Section~\ref{sec:model}.
We are given an acyclic directed network, with two special vertices 
called the origin $o$ 
and the destination $d$, respectively, such that each edge is on at least one $o$-$d$ path. Each edge $e$ of this
network is associated with an integer capacity $\mathpzc c_e$ and   a {positive} 
integer free-flow transit time $\mathpzc t_e$. 
Time horizon is infinite and discretized as $1,2,\ldots$. At each 
{time point}, 
a set of {selfish agents} 
enter the network from the
origin $o$, 
trying to reach the destination $d$ as quickly as possible. When an agent uses an edge $e$, two costs are
incurred to him: 
a fixed transit cost $\mathpzc t_e$ 
and a variable waiting cost dependent on
the volume of the traffic flow on  $e$ and the capacity $\mathpzc c_e$  as well as the agent's position in the queue of agents waiting 
{at} $e$. The total cost to each agent, also referred to as his \emph{latency}, is the
sum of the two costs on all edges he uses (which form an $o$-$d$ path in the network). The time {an} agent reaches the destination is simply his departure time
plus his latency.

{A critical setting} of our model is how the waiting time of each agent is determined in a way called
{\em deterministic queuing}. At each time step, for each edge $e$ of the network, there is a (possibly empty) queue of agents waiting at the tail of $e$, 
and at most a 
capacity number $\mathpzc c_e$ of agents who have the
highest priority (according to the queue order) can 
{start to move  along $e$}. After a fixed period of transit time $\mathpzc t_e$,
each of these agents reaches the head of $e$, and enters the next edge (waiting at its tail for traveling along it) unless the destination $d$ has been reached. The queue at
each edge is updated according to the  First-In-First-Out (FIFO) rule and pre-specified Edge Priorities -- if two agents enter an edge $e$ simultaneously,
{their queue priorities are determined by 
the priorities of the preceding edges from which they   enter $e$.} To break{
ties when agents} enter an edge simultaneously from the same edge, we divide each edge $e$ into
 a  
 number $\mathpzc c_e$ of lanes, and assign 
 {priorities among these lanes.} 
Initially, 
the agents who enter the {network} 
at the same time are associated with original 
priorities \emph{among them}, which are
{\emph{temporarily valid}} {only} when they enter the network. 
This is a 
{crucial} difference between our model and {that} in
{\cite{hm11,sst16}} where pre-specified  priorities \emph{among all agents} are {{\em permanent}}, working globally throughout the game. 

Each agent makes a {routing decision} at {\emph{every}} nonterminal vertex he reaches. {This is
the second {crucial and fundamental} difference of our model from almost all in the literature {\cite{hm11,sst16,whk14}}, in which
 agents are assumed to make their routing decisions \emph{only} at the  origin $o$ as to which $o$-$d$ path to take.
In addition, it is usually assumed that agents' decisions upon entering the system {are independent of} the
decisions of those agents who {are currently} in the system. {In game-theoretic
terminology}, while the dynamic traffic problem is usually modeled as a game of normal form (or a simultaneous-move game),
we model it as a game of extensive form, which allows for agents' \emph{online} adjustments and error corrections.

\subsection{Results}

Let $\Gamma$ denote our {\em extensive-form} game of dynamic traffic,  for which we apply the  solution concept of (pure strategy) {\em Subgame Perfect Equilibrium} (SPE). 
{The} analysis of SPE of $\Gamma$ is much more complicated than that of {{\em Nash Equilibrium}} (NE) of   the \emph{normal-form} game, denoted as $\Gamma^n$, which
corresponds  to the simplified setting where agents choose their one-off $o$-$d$ paths
simultaneously  at the very start of the game.
{We} relate our study of
{SPEs} for $\Gamma$ to that of {NEs} for $\Gamma^n$, and
establish four sets of results as specified below. 


\paragraph{\bf SPE existence in $\bfm\Gamma$.} We {demonstrate} that $\Gamma$ always possesses a very special SPE (which particularly implies the {\em NE existence} of $\Gamma^n$). 
{To the best of our knowledge, this is the first SPE existence result in the field of dynamic atomic routing games.}
Since typically $\Gamma$ is not finite and
multiple agents may
move simultaneously at each time {step}, the usual backward induction method does not work {here}.
Instead, we construct iteratively an intuitively simple SPE such that, {\em in any subgame of $\Gamma$}, the path profile  realized according to the SPE is {\em iteratively dominating} in the following sense: given dominating paths determined for a set $D$ of agents (initially none), the next dominating path determined for another agent $i$ (from his current location to the destination $d$) satisfies the property that as long as agents in $D\cup\{i\}$ follow their dominating paths, no agent outside $D\cup\{i\}$ can overtake $i$ at any vertex 
of $i$'s dominating path {(particularly the destination $d$)},
regardless of the choices of all the other agents. 

\paragraph{{\bf Equilibrium relationship between $\bfm{\Gamma^n}$ and $\bfm\Gamma$.}}
Given any NE in $\Gamma^n$, we construct an SPE of $\Gamma$ such that its realized path profile is the
given NE. 
This result along with {a} 
counterintuitive example shows
that {the NE set of $\Gamma^n$ is a {\em proper} subset of the SPE outcome set of $\Gamma$.} 
{On the one hand, {the proper containment reaffirms the} 
intuition that
model $\Gamma$ is more flexible than $\Gamma^n$. On the other hand, {the result builds} 
a useful bridge between the more frequently analyzed 
model $\Gamma^n$ and the more realistic but much more complicated model~$\Gamma$.}

\paragraph{\bf 
{NE properties} in $\bfm{\Gamma^n}$.} We show that in $\Gamma^n$ the set of {NEs} and that of {\em strong {NEs}} are identical, which particularly implies that all {NEs} of $\Gamma^n$ are {\em weakly Pareto Optimal}. The result is a corollary of the following stronger properties 
{for any NE} routing of $\Gamma^n$. Let agents be grouped according
to their arrival times at the destination~$d$.
\begin{mylabel5}
\item[-] {\em Hierarchal Independence}:  The arrival times at any vertex of agents in earlier groups are independent of {\em any} choices of agents in later groups (even if the {later} agents do not care about their own latencies and move in a
coordinated way).
\item[-]  {\em Hierarchal Optimality}: The equilibrium latency of any agent is minimum for him among all routings in which all agents in earlier groups follow their NE paths.
 \item[-] {\em Global {FIFO:}}
 vertex of the network (particularly the origin $o$), then the former will reach
the destination $d$ no later than the latter. 
\end{mylabel5}
{In contrast to the above global property,   FIFO in the literature} usually serves as an {\em assumption}
to regulate the queue on each {\em edge}. We would {also} like to emphasize that {the satisfications of all the above nice properites are by {\em all NEs} on {\em all networks},  while usually  in the related literature only} properties of some {\em special NEs} are studied.   


\medskip The above three sets of results remain valid even if the network has {\em multiple  origins} and a single destination.
This is one of the key reasons that we are able to prove the existence of an SPE for $\Gamma$ {(note that the off-equilibrium
scenarios are essentially new problems with multiple origins).}

\paragraph{\bf Bounded queue  {lengths}.}  Since {there may be infinitely many agents in {dynamic routing} models like ours,} 
  a natural question is whether the queue length at each
edge at equilibria can be bounded by a constant that 
depends {\em only} on the input {\em finite} network, when the inflow size
never exceeds the minimum capacity of an $o$-$d$ cut in the network. 
This interesting theoretical challenge has been listed as {an} 
important open problem in the {literature} 
\cite{cc15,sst16}.  {For a special} class of series-parallel networks, the so-called chain-of-parallel networks,  Scarsini et al. \cite{sst16} {proved} that the queue lengths at UFR equilibrium (i.e., NE with earliest arrivals) are bounded for   seasonal inflows.

We prove that equilibrium queue lengths   are indeed bounded in {our games on two}
{classes of networks: {$\Gamma^n$ on} series-parallel networks, and $\Gamma$ on networks {in which the} in-degree does not exceed the out-degree
at any nonterminal vertex.} {The {result of} bounded queue length} {implies} that the price of {anarchy} (PoA) and price of  stability (PoS) w.r.t. certain naturally defined social welfare measures  \cite{kp99,ss03} are also bounded  for the two classes of  {networks}.

\subsection{Related literature}

Classical models of dynamic flows, a.k.a. flows over {time,} 
were 
{pioneered}
by Ford and Fulkerson {\cite{ff62,ff58}, who} studied in a discrete-time setting the problem of maximizing the amount of goods that can be transported from 
{the origin to the destination} within a given 
{period of time.} 
{They showed that} this dynamic problem 
 {can be} reduced to 
 {the standard (static) minimum-cost flow problem and thus polynomially solvable.} 
{Philpott~\cite{p90} and Fleischer and Tardos~\cite{ft98} showed that many results can be carried over to the continuous counterpart.}
More recent {developments along this line of research} can be found
in \cite{fs07,s09}. 

Games of dynamic flows, {a.k.a.} routing games over time, were initially studied by \cite{v69,y71}, {which focused on analyzing Nash equilibria for} 
small-sized
concrete examples. {As summarized by \cite{ks09,pz01}, most subsequent studies} can be classified into four categories {in terms of methodologies}: mathematical programming,
optimal control, variational inequalities, and simulation-based approaches. 
{Variational} inequality formulations \cite{fbstw93} turn out {to be} the most successful in investigating Nash equilibria of the {\em simultaneous departure-time route-choice models}, where each player has to make a decision on his departure time 
 {as well.} Unfortunately,  little is known about the {existence}, uniqueness and characterizations of equilibria {under the general formulation \cite{mw10}.}

Our model is a {variant} of the
{\em deterministic queuing} (DQ) model which {was introduced} 
by \cite{v69}, developed by
\cite{hk81}, and recently revivified by \cite{ks09}. 
The 
notational and inconsequential differences are in the locations used to store queues -- the {standard} DQ model uses head parts of edges, while our storages, as well as those in \cite{cc15}, are located on tail parts. 
{The DQ model usually assumes that the waiting queue has no physical length, and the problem seems to be harder to analyze if  otherwise assumed \cite{d98}.} 


\paragraph{Non-atomic dynamic traffic games.}
Koch and Skutella~\cite{ks09,ks11} are the first to use a DQ model 
to study non-atomic dynamic flow games. 
They investigated the continuous-time  single-origin single-destination case with uniform inflow rates, i.e., the {\em temporal routing game} as named by \cite{bf15}. It was shown in \cite{ks11} that the Nash equilibria, i.e., Nash flows over time, can be 
viewed as concatenations of special static flows, the so-called {\em thin flows with resetting}. The 
Nash flow was characterized by the so-called {\em universal FIFO condition} that no flow overtakes another, 
and equivalently by an analogue to Wardrop principle that flow is only sent along currently shortest paths. 
Cominetti et al.~\cite{cc15} proved in a constructive way the existence and uniqueness of the Nash equilibrium of temporal routing games in 
a more general setting with piecewise constant inflow rates.  For multi-origin multi-destination cases, Cominetti et al.~\cite{cc15} proved in a {nonconstructive} way the NE existence when the inflow rates belong to the space of $p$-integrable functions with $1<p<\infty$. 

{Building on the temporal routing model and results of \cite{ks11},} Bhaskar et al.~\cite{bf15}   investigated the price of anarchy in Stackelberg routing games{, where the network manager acts as the leader picking a capacity for each edge within its given limit, and the agents act as followers, each of whom picks a path. 
The authors gave a polynomial-time computable strategy for the leader in an acyclic directed network to reduce edge capacities such that the PoA w.r.t. minimum completion time is $e/(e-1)$, and the PoA w.r.t. minimum total delay is upper bounded by $2e/(e-1)$, provided that the original network is saturated by its earliest arrival flow. Macko et al.~\cite{mls10} showed that the PoA of the temporal routing game w.r.t. to the minimum maximum delay can be as large as $n-1$ in networks with $n$ vertices.}
Anshelevich and Ukkusuri~\cite{au09} considered a 
dynamic routing game whose monotone increasing edge-delay functions are more general than DQ models but still obey the local FIFO principle. They studied how a single splittable flow unit present at each origin at time 0 would travel across the network to the corresponding destination, assuming that 
each flow particle is controlled by a different agent. {They showed that in} 
the single-origin single-destination case, 
there is a unique Nash equilibrium, {and the} 
efficiency of this equilibrium can be arbitrarily bad; {in} 
the multi-origin multi-destination case, however, {the} existence of an NE is not guaranteed. 

\paragraph{Atomic dynamic traffic games.} 
The most related work to ours  
is \cite{sst16}, {with several notable differences.} 
{First,} to break ties {Scarsini et al.~\cite{sst16} placed} 
priorities on agents 
rather than on edges. 
 {Second,} they only {allowed} 
 each agent to make a decision at the origin rather than at each {intermediary} vertex.} 
{Third,} {they focused} 
on seasonal inflows and how the transient phases impact the
long-run steady outcomes, whereas their notion of steady outcome does not apply in our model 
{because} the inflows
we consider are not restricted to be  {seasonal.} 
{Finally, they {concentrated} on a special kind of NE named  {\em Uniformly Fastest Route} (UFR) equilibrium, which is an NE 
{such that each agent} reaches all vertices 
{on} his route as early as possible.}
Using an argument similar to that in \cite{whk14}, Scarsini et al.~\cite{sst16} proved that the game admits at least one {UFR equilibrium.}  When the inflow is a constant not exceeding the capacity of the network, they characterized the flows and costs generated by optimal routings in general networks, and those generated by UFR equilibria in parallel networks and more general chain-of-parallel networks. 
The results on parallel networks were extended to the setting with 
seasonal inflows. 

In \cite{whk14}, more variants of {atomic games of dynamic traffic} {were} 
considered {for finitely many agents under discrete-time DQ models}. Apart from the sum-type
{of latency functions} as {considered} in \cite{sst16} and in this paper, Werth et al.~\cite{whk14} also studied the bottleneck-type
objective {functions}{ for agents, where the cost of each agent equals his expense on the slowest edge of his chosen path}. 
{To break ties,} {the global priorities placed on agents  as in \cite{sst16} were discussed for both the sum-objective and bottleneck-objective models, while the
local priorities placed on edges   as in this paper {were} investigated only for {the} bottleneck-objective model. Werth et al.~\cite{whk14} {focused} 
 on computational issues on NE and optimization problems, while we concentrate on SPE existence and NE {properties.}

As one of the earliest papers studying dynamic atomic routing games, \cite{hm09,hm11} {was} 
concerned with {computational} {complexity properties of NE and best-response strategies of a finite network congestion} model, where each edge of the network is viewed as a machine with a   processing speed {and a scheduling policy, and} each agent is viewed as a task   with a  {positive length which has to be processed by the machines one after another along the path the agent {chooses.} 
While the {positive} transit costs of agents are determined by machine speeds  and  {task lengths},  {their waiting costs are}  determined by {scheduling policies}. 
Apart from {the (local) FIFO policy (on each machine) that tasks are processed non-preemptively in order of arrival, the policy of (non)-preemptive global ranking and that of fair time-sharing were {also} investigated. For  agents (tasks) with uniform {lengths} in a single-origin directed network, when the FIFO policy is coupled with the global agent priorities for tie-breaking, the network congestion game turns out to be a generalization of the sum-objective model of  \cite{whk14}, and admit a strong NE which can be computed efficiently. Somewhat {surprisingly}, for the more restrictive case of single-origin single-destination networks, computing best responses is NP-hard. 
Kulkarni and Mirrokni~\cite{km15} studied the robust PoA of a generalization of the model in \cite{hm11} under the so-called {\em highest-density-first}   forwarding
policy.}

The global priority scheme introduced in \cite{fov08} was considered by  Harks et al.~\cite{hpsk16} for a dynamic route-choice game under the discrete-time DQ model without the FIFO queuing rule on each edge. 
They analyzed the impact of agent priority ordering on the efficiency of NEs for minimizing the total latency of all {$k$} agents, as well as their computability. 
They proved that the asymmetric game has its PoS in $\Omega(\sqrt k)$ and PoA upper bounded by $1+k^3/2$, while the symmetric game has its PoS and PoA equal to 1 and $(k+1)/2$, respectively. 
They also showed that an NE and each agent's best response are polynomially {computable}. 
Koch~\cite{k12} analyzed the dynamic atomic routing game on a restricted class of the discrete DQ model with global agent priorities, where each edge has zero free-flow transit time and unit capacity. 



\bigskip

The rest of this paper is organized as follows. In Section~\ref{sec:model} we present a formal definition
of our network game model along with its {extensive-form} setting $\Gamma$ and {normal-form} game setting $\Gamma^n$. In Section~\ref{sec:existence} we establish the existence of an SPE in $\Gamma$. In
Section~\ref{sec:NE-vs-SPE} we study {the} 
 {properties} of all {NEs} in $\Gamma^n$, and then investigate the relationship between the equilibrium flows of $\Gamma$ and $\Gamma^n$. In
Section~\ref{sec:bounded-queue} we bound the lengths of equilibrium queues for two special classes {of} networks. Finally, we conclude the paper with remarks on future research directions in Section~\ref{sec:conclusion}. All missing proofs and details are provided in the appendix.

\section{The Model}\label{sec:model}

In this section we formally present our network game model for dynamic traffic followed with some necessary notations and preliminaries. 
By edge subdivisions and duplications, we assume 
 {w.l.o.g.} in the rest of this paper that {\em each edge of the networks has a unit capacity  and a unit length}.



\subsection{{Dynamic traffic in directed networks}}\label{sec:dynamic}
 Let {$G$} be a finite {\em acyclic} directed {multi-graph} with two distinguished vertices, {the origin $o$ and the destination $d$},  such that every edge of $E$ is
on some $o$-$d$ path in $G$. The dynamic network traffic is represented by \emph{players} moving through $G$. The {infinite} time 
is discretized as $0,1,2,\ldots$ and at each {integer} time point $r\ge 1$, a {(possibly empty)} set $\Delta_r$ of {finitely many} players enter $G$ from
their common origin $o$  and each of them will go through some $o$-$d$ path in $G$ and leave $G$ from $d$. Each player,
when reaching a vertex {$v$ ($\ne d$)}, {\em immediately} selects an edge {$e$ outgoing from $v$} 
and enters $e$ {\em at once}. 
At any {\em integer} time point $r$, all players (if any) who have entered $e$  but not exited yet {\em queue}
at  {\em the tail part of~$e$}, and  only the {\em unique} head of the queue  {will} leave  the queue. {This queue head will {spend} one time unit in traversing $e$ from its tail to its head and exit $e$  at time $r+1$}. 

\paragraph{{{An} extended infinite network.}}  For ease of description, we make a technical extension of $G$ to accommodate all players of $\Delta:=\cup_{r\ge1}\Delta_r$ at the {very beginning.}
 Suppose that $|\Delta_r|$ is upper bounded by a fixed
integer {$F$ for all} $r\ge1$. We construct an infinite network {$\bar{G}$} from $G$ as follows:
add a new vertex $\bar{o}$, 
and connect $\bar{o}$ to $o$ with $F$ {internally} disjoint $\bar{o}$-$o$ paths:
{$P^f:=\bar{o}\cdots o^f_{r}o^f_{r-1}\cdots o^f_2o^f_{1}o$, $f=1,2,\ldots,F$. Each path $P^f$ has}
  a set of an infinite number of edges $\{o_{r+1}^fo_{r}^f\,|\, r\ge 0\}$, where
$o_0^f=o$, {and intersects} $G$ only at $o$. At time 0, all players of {$\Delta$ are queuing {at} 
and sets off from distinct edges in  $\cup_{f=1}^FP^f\subseteq\bar{G}$ (so they are all queue heads)} such that {for any $r\ge1$, players of $\Delta_r$ 
are on the tail parts of edges
$o^f_{r}o^f_{r-1}$, $f=1,\ldots,|\Delta_r|$, respectively; they are all at a distance $r$ from vertex $o$, and set off from these edges for their common destination $d$.}  
The 
{traffic} on
{$\bar G$} naturally corresponds to the one on {${G}$} in a way that the restriction of the {former} to its
subnetwork $G$ is exactly the {latter}---for any $r\ge1$, all players in $\Delta_r$ reach $o$ at
time $r$. 

For every $v\in \bar V$, we use $\bar E^+(v)$ and $\bar E^-(v)$  to denote the set of outgoing edges from   $v$ and the set of
incoming edges to $v$ {in $\bar G$}, respectively. For every $e\in\bar E$, we use {$u_e, v_e\in\bar V$} to denote the tail and head of $e$,
respectively, {i.e.,} $e=u_e v_e$.

\paragraph{Traffic regulation.}
A complete priority order $\prec_v$ is pre-specified over all edges in $\bar{E}^-(v)$ for any {$v\in \bar V$} {({$e_1\prec_v e_2$} means that $e_1$ has a higher priority than $e_2$).}
If players
queue  {at} 
(the tail part of) edge $e$, they are prioritized for entering and therefore exiting the queue according to the
following {\emph{queuing rule}}: For any pair of players, whoever enters $e$ earlier   
has {a} 
higher
priority; if they enter $e$ at the same time, then they must do so through two different edges of
$\bar{E}^-(u_e)$, in which case their {priorities are} determined by the priority order $\prec_{u_e}$ on the two edges.

\paragraph{Configurations.} We use $Q^r_e$ to denote the queue on edge $e$ at time $r$, which is both a sequence of players and the corresponding set. 
We call $c_r=(Q^r_e)_{e\in\bar E}$ a {\em configuration} at time $r$ if $Q^r_e\cap Q^r_{e'}=\emptyset$ for different $e$ and $e'$. Throughout this paper, \begin{mylabel5}
\item[-]$\Delta(c_r):=\cup_{ e\in \bar{E}}Q_e^r$ denotes the set of players involved in configuration $c_r$.
\item[-] $c_0$ denotes the unique {\em initial configuration} given by queues at time 0.
\end{mylabel5}Let $\mathcal C_r$ denote the set of configurations at time $r$.

\paragraph{Action sets.} Given configuration $c_r=(Q^r_e)_{e\in\bar E}$, the {\em action set} of  player $i\in Q_e^r$, {denoted} $\bar E(i,c_r)$,   is defined as follows: 
\begin{mylabel5}
\item[-] If 
 the head $v_e$ of $e$ is the destination $d$ and $i$ queues first in $Q_e^r$, then $\bar E(i,c_r):=\emptyset$, i.e., $i$ simply exits the system at time $r+1$;
 \item[-] If $v_e\neq d$ and $i$ queues first in $Q_e^r$, then $\bar E(i,c_r):= \bar E^+(v_e)$, i.e., player $i$ selects the next edge that is available at $v_e$; 
\item[-] Otherwise (i.e., $i$ is not the head of {$Q^r_e$}), player $i$ has to stay on $e$ with $\bar E(i,c_r):=\{e\}$. 
\end{mylabel5}


\paragraph{Consecutive configurations.} Given a configuration $c_r$ and an action profile $\bfm a=(a_i)_{i\in\Delta(c_r)}$ with $a_i\in \bar E(i,c_r)$, the traffic rule leads to a a new configuration $c_{r+1}= (Q_e^{r+1})_{e\in \bar{E}}$ at time $r+1$, referred to as a {\em consecutive configuration} of $c_r$:

 \begin{mylabel5}
 \item[-]  {As a set, $Q_e^{r+1}=\{i\in \Delta(c_r)\,|\, a_i=e\}$ consists of players choosing $e$ in {action profile ${\bfm a}$.}} 

 \item[-] {As a sequence, $Q_e^{r+1}$  
 equals $Q_e^r$ with its head removed 
followed by {$Q_e^{r+1}\sm Q_e^{r}$} 
      whose positions are according to the traffic regulation
      priority $\prec_{u_e}$ at the tail vertex $u_e$ of edge~$e$.} 
\end{mylabel5}

\subsection{Extensive form game setting $\Gamma$}


Given a dynamic traffic problem on network $\bar G$ with initial configuration $c_0$ (which is equivalent to the problem on  network $G$ with incoming flows $\{\Delta_r: r\geq 0\}$), {we study   a natural {extensive-form} game $\Gamma=\Gamma(\bar{G})$ as specified below.} Given any nonnegative integer $k$, we write $[k]$ for the set of positive integers at most $k$. Particularly $[0]=\emptyset$.

\paragraph{Histories.} For each time point $r\ge0$, 
 {a sequence of consecutive}  configurations $h_r=(c_0,c_1,\ldots,c_r)$ starting from the initial configuration $c_0$
is called a {\em history} at time $r$. The set of all possible
histories at time $r$ is denoted as~$\mathcal{H}_r$.

\paragraph{Strategies.} {Each player $i\in \Delta$ needs to make a decision 
at every configuration in  histories $h_r=(c_0,\ldots,c_r)$ with $ i\in \Delta(c_r)$.\footnote{Since $\bar G$ is acyclic, $i$ stays in $\bar G$ for a finite period of time, i.e., $i\not\in\Delta(c_s)$ for all $c_s$ with large enough $s$.}} The  {\em strategy}  of player $i\in \Delta$ is a mapping $\sigma_i$ that maps  each history $h_r=(c_0,\ldots,c_r)$ with $i\in\Delta(c_r)$ to an edge $\sigma_i(h_r)\in\bar E(i,c_r)$.
The {\em strategy set} of player $i$ is denoted as ${\Sigma}_{i}$.
A vector 
$\sigma=(\sigma_i)_{i\in\Delta}$ is called a {\em strategy profile} of $\Gamma$.

\paragraph{Latencies.}  Each strategy profile $\sigma$ of $\Gamma$ gives each player $i\in\Delta$ a latency equal to his exiting time, denoted as $t_i(\sigma)$, from the system, i.e., the length of time he stays in $\bar G$. Each player tries to minimize his {latency}, {which is equivalent to minimizing his travelling time from $o$ to~$d$.\footnote{Note that for any {time} $r$ and any player {in $\Delta_r$}, his arrival time {at} $o$ will always be $r$ under any strategy profile $\sigma$.}} 

\paragraph{Game tree.} The game tree of $\Gamma$ is typically an infinite tree with nodes corresponding to histories.  At each game tree node $h_{r}=(c_0,c_1,\ldots,c_{r})$, players in $\Delta(c_r)$ need to make decisions simultaneously, and their action profile leads to a new node $h_{r+1}$, which is a child of $h_{r}$.  
Each subtree of the game tree   rooted at history $h_r$ can be viewed as a separate game,  referred to as a {\em subgame} of $\Gamma$. The restriction of each strategy $\sigma_i$ to a subgame tree is also a strategy of the subgame. Given strategy profile $\sigma$ of $\Gamma$,  the time  when player $i\in \Delta(c_r)$ exits $\bar G$ in the subgame started from $h_r$ is denoted as  $t_i(\sigma|h_r)$.


\begin{definition}[{Subgame Perfect
Equilibrium}] A strategy profile $\sigma=(\sigma_i)_{i\in \Delta}$ is a {\em Subgame Perfect Equilibrium} (SPE) of $\Gamma$ if for any $r\ge0$ and any history $h_r\in\mathcal H_r$,  $t_i(\sigma|h_r)\leq t_i(\sigma'_i, \sigma_{-i}|h_r)$ holds
 for all  {$i\in \Delta(c_r)$} and $\sigma'_i\in \Sigma_i$, where $\sigma_{-i}$ is the partial strategy profile of players other than $i$.
\end{definition}



\subsection{Normal form game setting $\Gamma^n$}

A variant game model  $\Gamma^n$ is more popular in the literature and will also be studied in this paper.  {$\Gamma^n$ serves both as an independent model and
as a technical approach facilitating our analysis of the main model {$\Gamma$}.} 
{The key assumption in $\Gamma^n$ is that all players select an origin-destination path simultaneously at time 0.} They
should follow  the chosen paths and are not allowed to deviate at any vertex.
{All the other settings are the same as in game $\Gamma$.} 

We frequently analyze the {\em interim game} $\Gamma^n(c_r)$ of $\Gamma^n$ for each configuration $c_{r}=(Q^r_e)_{e\in \bar{E}}$ defined as follows.  The player set is $\Delta(c_r)$.  
For each player $i\in\Delta(c_r)$,
suppose $e_i(c_r)$ is the edge {at} 
which $i$ queues, i.e., $i\in Q^r_{e_i(c_r)}$, and $o_i(c_r)$ is the tail of $e_i(c_r)$. {The strategy set of player $i\in\Delta(c_r)$, denoted as ${\mathcal P}_i[o_i(c_r),d]$,  is the set of $o_i(c_r)$-$d$ paths containing $e_i(c_r)$.}

{For any player $i\in \Delta(c_r)$ and any path profile $\bfm p=(P_i)_{i\in \Delta(c_r)}$ with $P_i\in\mathcal P[o_i(c_r),d]$, we use $t_i^d(\bfm p)$ to denote the arrival time of $i$ at destination $d$ under the corresponding routing determined by $\bfm p$.}

\begin{definition}[{Nash  Equilibrium}] A path profile $\bfm \pi$ of $\Delta(c_r)$ is a {\em Nash Equilibrium} (NE) of $\Gamma^n(c_r)$ if no player can gain by uniliteral deviation, i.e., $t_i^d(\bfm \pi)\leq t_i^d(P_i,\bfm \pi_{-i})$ for all $i\in \Delta(c_r)$ and $P_i\in \mathcal P[o_i(c_r),d]$, where $\bfm \pi_{-i}$ is the partial path profile of players  {in $\Delta(c_r)$} other than $i$.\end{definition} 

\subsection{A warmup example}

 Each strategy profile $\sigma$ of $\Gamma$ induces a realized path profile, which is a strategy profile of $\Gamma^n$.  When $\sigma$ is an SPE of $\Gamma$, {Example
\ref{eg:SPEnotNE}} below shows that the induced path {profile}  
is not necessarily an NE of $\Gamma^n$, demonstrating that richer phenomena may be observed in $\Gamma$ than in $\Gamma^n$. {In our illustrations, we only show the
{restrictions} of games $\Gamma$ and $\Gamma^n$ to $G$, which we denote as $\Gamma(G)$ and $\Gamma^n(G)$ respectively. Accordingly, the {{\em travelling cost} a player spends} in $G$ from $o$
to $d$ is his arrival time at $d$ minus that at $o$. }

\begin{example}\label{eg:SPEnotNE}
 The network   $G$ is as depicted in Figure \ref{fig:SPEnotNE}. 
Edge $e_1$ (resp. $e_2$) has a higher priority than  $e_3$ (resp. $e_4$). There are only two players in the game, who enter $G$ via origin~$o$ at the same time $r=1$. Player~$1$ has
a higher original priority than player $2$. In $\Gamma$, player $1$ takes a vicious strategy in the following sense.  He initially chooses edge {$ov$} and then tries to
block player 2 by choosing {$v\rightarrow w_1\rightarrow d$} if  player $2$ used edge  {$ou_1$} and  {$v\rightarrow w_2\rightarrow d$}
otherwise. Player $2$ always follows  {$o\rightarrow u_1\rightarrow w_1\rightarrow d$}. It is easy to check that {these yield a strategy profile that is
an SPE of $\Gamma(G)$ (off-equilibrium behaviors of the two players can be easily defined), {incurring a travelling} cost 3} for player 1 and 4 for player 2. However,  the induced path
profile, {$ovw_1 d$} for player 1 and {$ou_1w_1d$} for player $2$, is not an NE of {$\Gamma^n(G)$. Indeed, $\Gamma^n(G)$ admits in total six
NEs, all bringing} 
the two players the same {travelling} cost of 3.
\end{example}

\begin{figure}[h!]
  \centering
  \includegraphics[width=6cm]{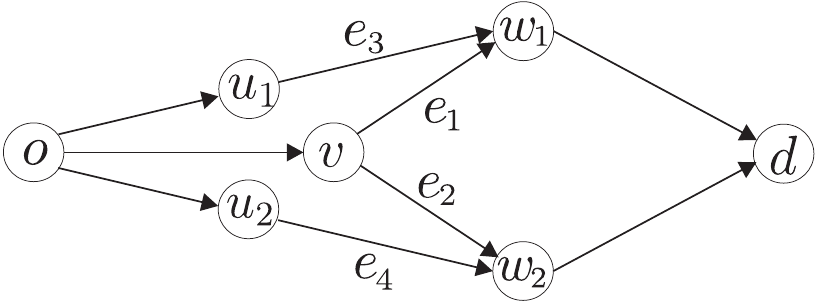}\\
  \caption{{An SPE of $\Gamma$ may not induce an NE of $\Gamma^n$.}}\label{fig:SPEnotNE}
\end{figure}

\section{Existence of SPE}\label{sec:existence}

In this section, we establish the existence of an SPE for game $\Gamma=\Gamma(\bar G)$ by using an algorithmic 
method.  Given any feasible configuration $c_r$, our algorithm (Algorithm \ref{proc:id} below) computes a path profile, which is a special NE, of $\Gamma^n(c_r)$. The special NEs computed for all configurations will assemble an SPE of $\Gamma$ in a way that for each configuration $c_r$ the path profile induced by $c_r$ and the SPE is exactly the NE computed for $\Gamma^n(c_r)$.  

The high-level idea behind our  algorithm is a natural greedy best-response strategy: iteratively computing the unique best path {based} on all paths that have been obtained. Similar approaches with noticeable differences have been adopted to deal with various routing games, e.g., 
\cite{hm11,sst16,whk14}. The special difficulties in our setting are imposed by 
subgames that start   from any feasible  configurations (which are essentially   multi-origin problems) and priorities that are placed on edges rather than on players. 

For  any $S\subseteq \Delta(c_r)$, we also use
$\bfm p=(P_i)_{i\in S}$ with $P_i\in\mathcal P[o_i(c_r),d]$
{to denote} a partial path profile, corresponding to a (partial) routing. Given any vertex $v\in \bar V$, edge $uv\in\bar E$, and  player $i\in S$, 

\begin{mylabel5}
\item[-]  {$t^{v}_i(\bfm p)$ (resp. $t^{uv}_i(\bfm p)$) denotes the time when $i$ reaches  $v$ (resp. enters  $uv$)  under   $\bfm p$;}
\item[-]  {$t^{v}_i(\bfm p):=\infty$ if $v\not\in P_i$, and $t^{uv}_i(\bfm p):=\infty$ if $uv\notin P_i$.}
\end{mylabel5}
{If $S'\subseteq S$, we often write $(P_i)_{i\in S'}$ as $\bfm p_{S'}$.}


 Given a feasible configuration $c_r$,  
our algorithm goes roughly as follows.  Initially, let player subset $D$ of $\Delta(c_r)$ and partial routing $\bfm\pi_D$ of players in $D$ be empty. 
Then recursively, assuming players in $D$   go along  their {\em selected} paths specified in $\bfm\pi_D$, 
we {enlarge $D$ with  a new player $i \in \Delta(c_r)\sm D$ and an associated path $\Pi_{i}\in\mathcal P[o_i(c_r),d]$ in the following way. For each $j \in \Delta(c_r)\sm D$, define $\tau_{j}^v=\min\{t^v_j(\bfm\pi_D,R_j)\,|\,R_j\in\mathcal P[o_j(c_r),d]\}$ as the {``ideal latency''} of player $j$ at {vertex} $v$ w.r.t. $\bfm\pi_D$.}  {Let $i\in \Delta(c_r)\sm D$ and  $\Pi_{i}\in\mathcal P[o_i(c_r),d]$ be such that (i) $t^d_i(\bfm\pi_D,\Pi_i)$ equals the smallest  {``ideal latency''} at $d$ among  all players in $\Delta(c_r)\sm D$, i.e., $t^d_i(\bfm\pi_D,\Pi_i)=\min_{j\in \Delta(c_r)\setminus D}\tau_{j}^v$, (ii) if there are more than one of such {candidates $(i,\Pi_i)$}, then {the choice of $(i,\Pi_i)$ is made such that} $\Pi_i$ has the highest priority w.r.t. $\prec_d$, (iii) if still more than one {candidates} satisfy (ii), in which case {the paths involved} must share the same ending edge and suppose $u$ is the tail vertex of this edge,  then $t^u_i(\bfm\pi_D,\Pi_i[o_i(c_r),u])$ is the smallest. Repeat the above process until $o_i(c_r)$ is reached, in which case all the remaining players must all {queue at}
$e_i(c_r)$ and $i$ is the head of this queue.} The process is repeated while $D$ and $\bfm\pi_D$ {become} larger and larger. A formal description of the process is  {presented in Algorithm \ref{proc:id}}. 


 \begin{algorithm} \SetAlgoNoLine \KwIn{a feasible configuration $c_r$ at time $r$.}
\KwOut{the iterative dominating path profile (routing) $\bfm\pi=(\Pi_i)_{i\in\Delta(c_r)}$ for $\Delta(c_r)$ along with the corresponding player indices 1,2,\ldots.}
\begin{mylabel5}
\item[\,\,\,1.] 
Initiate $D\leftarrow\emptyset$,
          $\bfm\pi_{[0]}\leftarrow\emptyset$, $i\leftarrow0$.
\item[\,\,\,2.]\vspace{-1mm} {$i\leftarrow i+1$ (NB: Start to search for a new dominator and his dominating path).}
\item[\,\,\,3.]\vspace{-1mm} 
For each player $j\in \Delta(c_r)\sm D$ and vertex $v\in\bar G$,
\begin{mylabel4}
\item[-] let $\tau_{j}^v=\min\{t^v_j(\bfm\pi_{[i-1]},R_j)\,|\,R_j\in\mathcal P[o_j(c_r),d]\}$
          be the earliest time for $j$ to reach vertex $v$ from his {\em current location in $c_r$}, assuming
          that \emph{all  other} players in {$\bar G$} are those in $D$ and they    go along their paths specified in $\bfm\pi_{[i-1]}$;
          \item[-] let $\mathcal P_{j}^v$ denote the set of all the \emph{corresponding} $o_j(c_r)$-$v$ paths for $j$ to reach $v$ at time          $\tau_{i}^v$. 
          \item[-] if there is no 
              {such path,} then set $\tau_j^v\leftarrow\infty$ and $\mathcal P_j^v\leftarrow\emptyset$.
          \end{mylabel4}

\item[\,\,\,4.]\vspace{-1mm} $w\leftarrow d$, $D'\leftarrow \Delta(c_r)\sm D$, $\Pi\leftarrow\emptyset$. (NB: Steps~3--10 are to help
          identify player {$i$ and path $\Pi_i$} in Steps~11 and 12; $D'$ holds the candidates for $i$; $\Pi$ is a subpath of $\Pi_i$ that will grow edge by edge starting from $d$.)
\item[\,\,\,5.]\vspace{-1mm} \textbf{While}  $\tau:=\min_{j\in D'}\tau_j^w \ge r+1$  \textbf{do}
\item[\,\,\,6.]\vspace{-1mm} \hspace{2mm} $D'\leftarrow\{j\in D'\,|\,\tau_{j}^w=\tau\}$;
\item[\,\,\,7.]\vspace{-1mm} \hspace{2mm} let $uw$ be the edge of highest priority  among all last edges of paths in $\cup_{j\in D'}\mathcal P_j^w$;
\item[\,\,\,8.]\vspace{-1mm} \hspace{2mm} $\Pi\leftarrow \Pi\cup\{uw\}$;
\item[\,\,\,9.]\vspace{-1mm} \hspace{2mm} $w\leftarrow u$;
\item[10.]\vspace{-1mm} \textbf{End-While} (NB: at the end of the while-loop all players in $D'$ are queuing on the starting edge of $\Pi$ in configuration $c_r$)
\item[11.]\vspace{-1mm} Let $i\in D'$ be the player who stands  first (among all players in $D'$) in line on the starting edge of $\Pi$ in configuration $c_r$. 
\item[12.]\vspace{-1mm} Let $i$ {\em select} $\Pi_i\leftarrow\Pi$.
\item[13.]\vspace{-1mm}   $D\leftarrow D\cup \{i \}$, $\bfm\pi_{[i]}\leftarrow(\bfm\pi_{[i-1]},\Pi_i)$. (NB: the algorithm outputs player $i$ and his path $\Pi_i$.)
\item[14.]\vspace{-1mm} Go to Step 2.
\end{mylabel5}
\caption{\sc(Iterative Dominating Path Profile)}\label{proc:id}
\end{algorithm}

{It is worth noting that, for each $i\ge1$, Algorithm \ref{proc:id} can identify player $i$ and path $\Pi_i$ in finite time via ignoring players and vertices that are sufficiently far from 
{the origin $o$} in configuration $c_r$.}



Suppose Algorithm \ref{proc:id} indexes  the players of  $\Delta(c_r)$  as $1,2, \ldots$ with the associated path profile $\bfm\pi=(\Pi_1,\Pi_2,\ldots)$.    As {can be seen from Lemma \ref{lem:ineq} below, each player $i\in\Delta(c_r)$ is a {\em dominator} in $\Delta(c_r)\sm[i-1]$ and $\Pi_i$ is a {\em dominating path} w.r.t. $\bfm\pi_{[i-1]}$ in the following sense: under the assumption that players in $[i-1]$ all follow $\bfm\pi_{[i-1]}$, as long as $i$ takes $\Pi_i$, he will be among
the  first {(within $\Delta(c_r)\sm [i-1]$)} to reach the destination {$d$} and no player in $\Delta(c_r)\sm [i]$ will be able to reach any of {$\Pi_i$'s} intermediary vertices earlier
than {$i$} does, regardless of the choices of players in $\Delta(c_r)\sm [i]$.
We may also call the above $\bfm\pi=(\Pi_1,\Pi_2,\ldots)$ an {\em iterative dominating path profile}.}

\begin{lemma}\label{lem:ineq}
{Given any feasible configuration $c_r$ at time $r$, let   players $1,2,\ldots$ of $\Delta(c_r)$ be as  indexed and path profile $\bfm\pi=(\Pi_i)_{i\in\Delta(c_r)}$ be as computed in Algorithm \ref{proc:id}. If players  $i,j\in \Delta(c_r)$ and player subset $S$ satisfy $j\in S\subseteq\Delta(c_r)\sm[i-1]$, then for any vertex $v\in\Pi_i$ and any path profile $\bfm p$ of $\Gamma^n(c_r)$, it holds that
\[
t_i^v(\bfm\pi_{[i]}, \bfm p_{S\,\sm\,\{i\}})=\min\{t_i^v(\bfm\pi_{[i-1]},R_i)\,|\,R_i\in\mathcal P[o_i(c_r),d]\}\le t_j^v(\bfm\pi_{[i-1]},
\bfm p_S).
\]}
\end{lemma}

{The iterative dominance of $\bfm\pi$  particularly implies that given the selected paths in $\bfm\pi_{[i-1]}$ of players in $[i-1]$, player  $i$ has no incentive to deviate from $\Pi_i$ for each $i$, saying that $\bfm\pi$ is an NE of $\Gamma^n(c_r)$. Furthermore,  a special SPE of $\Gamma$ can be constructed, building on the iterative dominating path profiles for all feasible configurations.}



\begin{theorem}\label{thm:SPE}
Game {$\Gamma$} admits an SPE.
\end{theorem}
\begin{proof}Given any history $h_r=(c_0,\ldots,c_r)\in\mathcal H_r$ for any time point $r\ge 0$,  suppose the players in $\Delta(c_r)$ are named as $1_r,2_r,\ldots$ such that player $i_r$ is the $i$th player {added to $D$} in Step~13 of Algorithm~\ref{proc:id} (with input $c_r$). For each player $i_r\in\Delta(c_r)$, let {$\Pi^{c_r}_{i_r}$ denote the path selected by $i_r$ in Algorithm~\ref{proc:id}; recall that under $c_r$ player $i$ {queues at} 
the first edge   of $\Pi^{c_r}_{i_r}$.}
A configuration in $\mathcal C_{r+1}$ will {result from} $c_r$ according to action profile {$\bfm a^{c_r}$} defined as {follows:
\[\text{The action of }i=\left\{\begin{array}{ll}\text{the first edge of }\Pi^{c_r}_{i_r},&\text{if $i$ queues after someone else;}\\
\emptyset,&\text{if $i$ queues first on the last edge of }\Pi^{c_r}_{i_r};\\
\text{the second edge of }\Pi^{c_r}_{i_r},&\text{otherwise}.\end{array}\right.\]
The set {$\cup_{r\ge0}\cup_{c_r\in\mathcal C_r}\bfm a^{c_r}$}} of action profiles  defines a strategy profile $\sigma^*=(\sigma^*_i)_{i\in\Delta}$ of {$\Gamma$}. We will prove that {$\sigma^*$} is {an} SPE of $\Gamma$.

Let $(c_r,c_{r+1},c_{r+2},\ldots)$ be the list of configurations and {$(P^*_i)_{i\in\Delta(c_r)}$} be the path profile induced by {$h_r$} and $\sigma^*$. It can be deduced from Lemma \ref{lem:ineq} and {Algorithm \ref{proc:id}} that
\begin{mylabel}
\item[$\bullet$] For any $s\ge r+1$, player sequence $1_s,2_s,\ldots$ is a subsequence of $1_{s-1},2_{s-1},\ldots$ such that $\Delta(c_{s-1})\sm\Delta(c_{s})$ consists of the first $|\Delta(c_{s-1})\sm\Delta(c_{s})|$ players of $1_{s-1},2_{s-1},\ldots$
\item[$\bullet$] For any $s\ge r+1$ and $i\in\Delta(c_s)$, {$\Pi^{c_s}_{i}\subseteq\Pi^{c_{s-1}}_i$.}
\end{mylabel}
Therefore, the path {$P^*_i$} formed by the actions of each player $i\in\Delta(c_r)$ is exactly {$\Pi^{c_r}_i$}. According to Lemma \ref{lem:ineq}, we have, for any {$i\geq 1$,
\begin{center}
$\displaystyle t_{i_r}(\sigma^*|h_r)=\min\{t^d_{i_r}(\Pi^{c_r}_{1_r},\ldots,\Pi^{c_r}_{(i-1)_{r}}, P_{i_r})\,|\,P_{i_r}\in\mathcal P[o_{i_r}{(c_r)},d]\}$ {for each $i\ge1$.}
\end{center}
 Moreover,  for any $j\ge 1$ and any strategy profile $\sigma'$ of $\Gamma$ with $\sigma'_{i_r}=\sigma^*_{i_r}$} for all $i\in[j]$, the path profile $(P'_i)_{i\in\Delta(c_r)}$ induced by {$h_r$} and $\sigma'$ satisfies {$P'_{i_r}=\Pi^{c_r}_{i_r}$} for all $i\in[j]$.

Now given any $k\ge 1$ and any $\sigma_{k_r}'\in\Sigma_{k_r}$, we consider strategy profile  {$\sigma'=(\sigma_{k_r}',\sigma^*_{-k_r})$} and the path profile  {$\bfm p'=(P'_{i})_{i\in \Delta(c_r)}$} induced by {$h_r$} and $\sigma'$. We have {$P'_{i_r}=\Pi^{c_r}_{i_r}$} for all $i\in[k-1]$, and
\[t_{k_r}(\sigma'_{k_r},\sigma^*_{-k_r}|h_r)=t_{k_r}^d(\bfm p')=t_{k_r}^d(\Pi^{c_r}_{1_r},\ldots,\Pi^{c_r}_{(k-1)_r},\bfm p'_{\Delta(c_r)\,\sm\,\{1_r,\ldots,(k-1)_r\}}).\]
It follows from Lemma \ref{lem:ineq} {that
\begin{center}$t_{k_r}(\sigma'_{k_r},\sigma^*_{-k_r}|h_r)\ge \min\{t^d_{k_r}(\Pi_{1_r},\ldots,\Pi_{(k-1)_{r}}, P_{k_r})\,|\,P_{k_r}\in\mathcal P[o_{k_r}(c_r),d]\}=t_{k_r}(\sigma^*|h_r)$.\end{center}
 The} arbitrary choices of $k$ and $\sigma'_{k_r}$ imply that {$\sigma^*$} is an SPE of $\Gamma$.
\end{proof}

We would like to remark that 
placing  priorities on edges is crucial to 
the existence of an SPE. As the following example shows, {the guarantee of SPE existence would be impossible if priorities were placed on players.} 

\begin{example}Consider a subgame starting with a configuration illustrated in Figure~\ref{noSPE}. There are {9 players in total} represented by small rectangles on the {edges}, with players $i,j,k$ being our focus (the remaining 
{6 players} do not need to make substantial decisions). 
The global priorities placed on the players are such that $i$ ranks higher than $j$ and $j$ higher than $k$. However, $k$ reaches $v_1$  {(or $v_2$)} one time unit earlier than $i$ (if they choose to pass the same vertex) and adds $1$ to the waiting time of $i$. Therefore, $i,j,k$ have a Rock-Paper-Scissors like relationship. It can be checked that this subgame does not have any NE. Since each of the three players $i,j,k$ only needs to make one decision, it follows that an SPE does not exist in this subgame.  
\end{example}

\begin{figure}[h!]
  \centering
  \includegraphics[width=12cm]{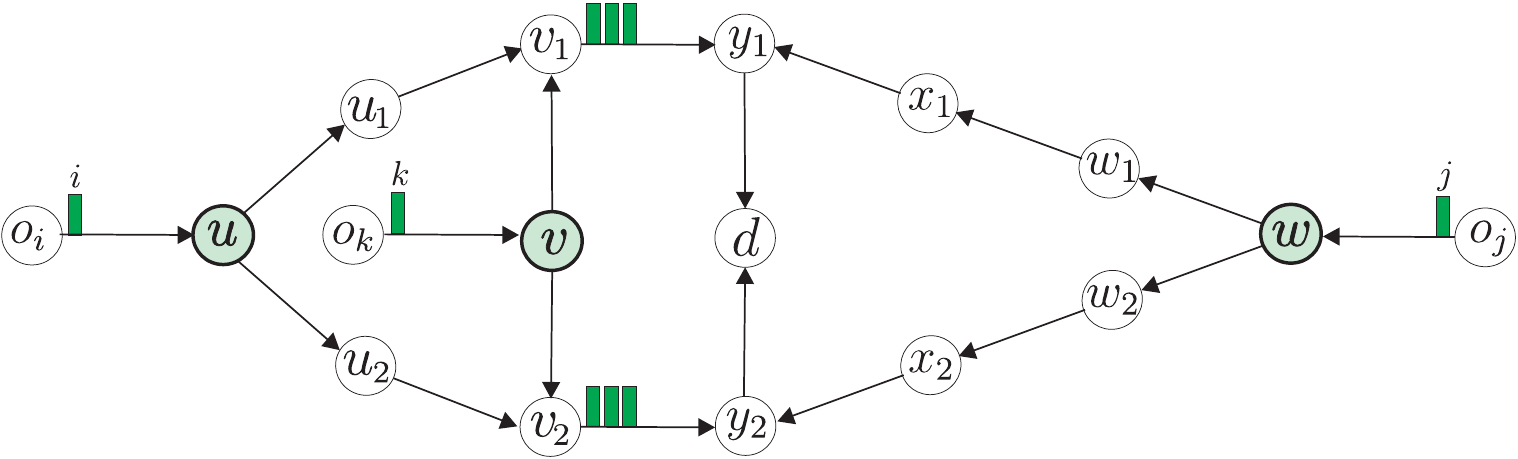}\\
  \caption{SPE 
  {might} not exist if priorities 
  {were} placed on players. In the subgame, player $i$ ranks higher than player $j$, and $j$ ranks higher than player $k$.}\label{noSPE}
\end{figure}

In contrast, the subgame in  Figure \ref{noSPE} will have an NE if priorities are placed on edges as in our model. Suppose, e.g.,  $v_1y_1\prec_{y_1}x_1y_1$ and $x_2y_2\prec_{y_2}v_2y_2$, then it can be seen that the following implies an NE:  $k$ chooses path $o_kvv_1y_1d$, $j$ chooses path $o_jww_2x_2y_2d$,  and $i$ makes an arbitrary choice. If in addition $y_2d\prec_d y_1d$, then \[(v_2y_2d,v_1y_1d,v_2y_2d,v_1y_1d,v_2y_2d,v_1y_1d,o_jw w_2 x_2y_2d,o_kvv_1y_1d,o_iuu_2v_2y_2d)\] is the iteratively dominating path profile for the configuration.

\section{Equilibria of {interim games}} 
\label{sec:NE-vs-SPE}

Being  natural and reasonable, models that are quite similar to $\Gamma^n$ have been the focus in the {studies of}
dynamic traffic games \cite{hm09,sst16,whk14}. While only special {NEs} were
studied in {the literature}, we {establish some general properties of all NEs of} $\Gamma^n$.
Using these properties, we prove that the NE outcomes of game $\Gamma^n$  form a proper subset of the SPE outcomes of game $\Gamma$.
Therefore, $\Gamma^n$ can serve not only as an independent model but also
as a technical approach that facilitates our analysis of the main game model $\Gamma$.

Given any configuration $c_r$ at time $r$, when studying {some} special
SPE of game $\Gamma$ in Section~\ref{sec:existence}, we obtain an iterative dominating path profile{,} which is a
{\em special} NE of the interim game $\Gamma^n(c_r)$, where players 
in $\Delta(c_r)$ are {\em completely  ordered}
such that 
{the ones} with {smaller indices}
have advantages over those with {larger indices}.
For studying  {\em every} NE of $\Gamma^n$, it turns out that batching players according to their arrival times at the destination {$d$}
is useful.
{For any} NE $\bfm\pi$  of $\Gamma^n(c_r)$, {let} 
$\tau(\bfm\pi,1)<\tau(\bfm\pi,2)<\tau(\bfm\pi,3)<\cdots$ be the arrival times of all players in $\Delta(c_r)$ at $d$ under $\bfm\pi$. For each {integer} $k\ge1$, let
 \[\Delta(\bfm\pi,k):=\{i\in{\Delta(c_r)}\,|\,t^d_i(\bfm\pi)=\tau(\bfm\pi,k)\}\] be the set of players in $\Gamma^n(c_r)$ who {reach} $d$ at {the $k$th earliest time}, i.e., $\tau(\bfm\pi,k)$ under $\bfm\pi$; {we often refer to $\Delta(\bfm\pi,k)$ as the {\em $k$th batch}. We use}
\[\Delta(\bfm\pi,{[k]}):=\cup_{j\in[k]}\Delta(\bfm\pi,j)\] {to denote} {the set of players} reaching $d$ no later than $\tau(\bfm\pi,k)$, {i.e., those in the first $k$ batches}. {For notational convenience, we set $\Delta(\bfm\pi,[0]):=\emptyset$ {to be the $0$th batch, and let} $\Delta(\bfm\pi,[\infty]):=\Delta(c_r)$ {denote the disjoint union of all batches}.}





\subsection{A motivating example}
In studying NE of $\Gamma^n$, we need to frequently analyze what happens if one player unilaterally {deviates} by choosing a different path. 
This is a quite tricky issue {in general. Players} may affect {one another}
in quite unexpected ways. As the following 
 {extreme} example shows, adding a player to the system may
weakly improve the performances of other players {(nobody is worse off and at least one is better off)}, or equivalently, removing a player may weakly harm the performances of other players.

\begin{example}\label{addaplayer}
The configuration $c_r$ at time $r$ is  {as depicted} in  Figure \ref{better}. There are 10 players (shown as small rectangles on edges) in total, with $i,j,k$ being our focus (note that Figure \ref{better} {shows} only a part of the whole network). 
Edge $e_1$ has a higher priority than  $e_2$. Let  player $k$ {choose the top} 
path $u_1u_2u_3u_4u_5d$, player $i$ {choose} the middle path  $u_1u_2v_3v_4d$, and player $j$ {follow the bottom} 
path $v_1v_2v_3v_4v_5d$. It can be checked that all the three players reach $d$ at time $r+5$ (and the three chosen paths  along with the trivial paths of the {other 7} players constitute an NE for $\Gamma^n(c_r)$).  Suppose now player $k$ is removed from the system, and $i$ and $j$ keep their chosen paths. This removing  makes $i$ arrive at vertices {$u_2$, $v_3$ and $v_4$} one unit of time earlier. Since $i$ {has} to pay one  extra {unit of} waiting cost on edge $v_4d$, {he} reaches $d$ still at time $r+5$. However, the earlier arrival of $i$ at $v_3$ delays $j$, making {him} reach $d$ at time $r+6$. {Note that both path profiles are NEs.}
\end{example}
\begin{figure}[h!]
  \centering
  \includegraphics[width=8cm]{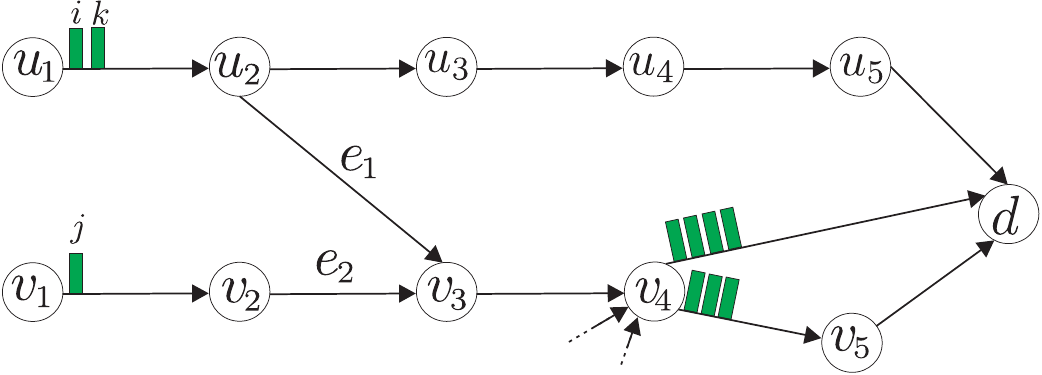}\\
  \caption{{Removing a player (weakly) harms the performances of other players}.}\label{better}
\end{figure}

\subsection{A critical lemma}

Throughout this subsection, we {fix} a strategy profile $\bfm\alpha=(A_i)_{i\in\Delta(c_r)}$ of $\Gamma^n(c_r)$
and {consider} a fixed player $\zeta\in\Delta(c_r)$ {and the scenario} where only player $\zeta$ changes his strategy. Let {$\bfm\alpha_{-\zeta}=(A_i)_{i\in\Delta(c_r)\setminus\{\zeta\}}$} denote the partial strategy profile of players in $\Delta(c_r)$ other than $\zeta$.
For each vertex $v$, define
$$
\tau^v:=\min_{A'_{\zeta}\in \mathcal P[o_\zeta(c_r),d]}\{t_{\zeta}^v(A'_{\zeta},\bfm{\alpha}_{-\zeta})\}
$$
as the earliest time at which player $\zeta$ can {reach} vertex $v$ by unilaterally changing his {path} 
(if there is no path from $o_\zeta(c_r)$ to $v$, then {set} $\tau^v:=+\infty$). {Analogously, for each player} {$j\in\Delta(c_r)\sm\{\zeta\}$} and vertex $v\in A_j$, define
$$\tau_{j}^v:=\min_{A'_{\zeta}\in \mathcal P[o_\zeta(c_r),d]}\{t_{j}^v(A'_{\zeta},\bfm{\alpha}_{-\zeta})\}$$ {as the earliest} time at which player $j$ can {reach} vertex $v$ when  player $\zeta$ unilaterally {changes his path. We emphasize that $j$ keeps his path at $A_j$ in the definition of $\tau_{j}^v$.}

In the following, for any non-degenerated path $P$ in $\bar G$ and vertex $v\in P$, we use $e_v(P)$ to {denote} the
edge with head $v$ in $P$.

\begin{definition}
For every player  {$j\in\Delta(c_r)\sm\{\zeta\}$} and vertex $v\in A_j$, we say player $\zeta$ {\em dominates} player $j$ at vertex $v$ under $\bfm\alpha$ {if} 
\begin{mylabel}
\item  $\tau^v<\tau_j^v$; or
\item $\tau^v=\tau_j^v$,  and   {there exists an} $o_\zeta(c_r)$-$v$-$d$ path $A^*_\zeta$ such that $t_{\zeta}^v(A^*_{\zeta},\bfm{\alpha}_{-\zeta})=\tau^v$ and $e_v(A^*_\zeta)\preceq_v e_v(A_j)$.
\end{mylabel}
\end{definition}

In words, player $\zeta$ dominates $j$ at $v$ if and only if either the earliest possible time that $\zeta$ reaches $v$ is earlier than the earliest possible time that $j$ reaches $v$, or the two earliest arrival times are equal and $\zeta$ has a corresponding path that allows {him} to enter $v$ from an edge with a higher or equal priority over the edge taken by $j$,  assuming that only $\zeta$ is allowed to deviate his path from the given path profile $\bfm\alpha$.

\begin{lemma}\label{dominating}
For any player  {$j\in\Delta(c_r)\sm\{\zeta\}$} and vertex $v\in A_j$, if there exists a path $A'_\zeta\in \mathcal P[o_\zeta(c_r),d]$ such that $t_j^v(A'_\zeta,\bfm{\alpha}_{-\zeta})\neq t_j^v(\bfm{\alpha})$,
i.e., player $\zeta$ can influence player $j$'s arrival time {at} $v$ by unilaterally changing his {path,} 
{then $\zeta$ dominates $j$ at $v$}.
\end{lemma}

The above lemma implies in particular that if player $\zeta$ can influence  the arrival time of player $j$ at 
{vertex} $v$, then $\zeta$ is able to {reach} $v$ earlier than $j$ {does}. Note that this result would be impossible if priorities were placed on players{, which can be seen from Figure~}\ref{noSPE}. We focus on player $k$ and suppose the path profile $\bfm{\alpha}$ is such that $i$ and $j$ both choose their upper paths and $k$ chooses his lower path. Then $k$ is able to influence player $j$'s arrival time at node $y_1$ by deviating to the upper path. However, $k$ is unable to {reach} $y_1$ earlier than $j$ {does}.

{From the  proof of Lemma \ref{dominating}, we also have the following stronger result.}

\begin{corollary}\label{alldominating}
If player $\zeta$ dominates some player {$j\in \Delta(c_r)\sm\{\zeta\}$} at vertex $v\in A_j$, then $\zeta$ dominates $j$ at all the other vertices on the subpath $A_j[v,d]$.
\end{corollary}

As an immediate result of Corollary \ref{alldominating}, we have the following lemma.
\begin{lemma}\label{nodominating}
Let $\bfm{\pi}=(\Pi_i)_{i\in{\Delta(c_r)}}$ be an NE of $\Gamma^n(c_r)$. For every $k\ge1$ and   every player $j\in\Delta(c_r)\sm\Delta(\bfm{\pi},[k])$, $j$ {cannot} dominate any player $i\in\Delta(\bfm{\pi},[k])$ at any vertex of path $\Pi_i$.
\end{lemma}
\begin{proof}
Suppose on the contrary that there does exist a player  $j\in\Delta(c_r)\sm\Delta(\bfm\pi ,[k])$ and player $i\in\Delta(\bfm{\pi},[k])$ such that $j$ dominates $i$ at some vertex $v\in\Pi_i$.  Then by Corollary \ref{alldominating}, $j$ also dominates $i$ at vertex $d$. This means that there exists a path $\Pi_j^*$ such that  $$t_j^d(\Pi_j^*,\bfm{\pi}_{-j})\le\min_{\Pi_j'\in \mathcal P[o_j(c_r),d]}\{t_i^d(\Pi_j',\bfm{\pi}_{-j})\}\le t_i^d(\bfm{\pi})=\tau(\pi,k).$$  However, {$t^d_j(\bfm\pi)>\tau(\bfm\pi ,k)$ due to $j\in\Delta(c_r)\sm\Delta(\bfm\pi ,[k])$, indicating that $j$ has an incentive to deviate to 
{$\Pi_j^*$} and hence violating the fact that $\bfm\pi$ is an NE.} 
\end{proof}

\subsection{Computing the earliest-arrival best response}

{For any configuration $c_r$ and the corresponding 
{normal-form game} $\Gamma^n(c_r)$, suppose now we are given a path profile $\bfm\alpha$. We show in this subsection how to compute a special best-response defined below.}

\begin{definition}
Given a partial path profile $\bfm\alpha_{-\zeta}$, an $o_\zeta(c_r)$-$d$ path $P^*$ of player $\zeta\in\Delta(c_r)$ is called the {\em earliest-arrival best response} if (i) $\zeta$ {arrives} at each node of $P^*$ the earliest among all  $o_\zeta(c_r)$-$d$ paths, and (ii) when there are more than one earliest ways to arrive at a node of  $P^*$, $P^*$ always selects an entering edge with the highest 
{priority}.
\end{definition}

{We shall see  that the earliest-arrival best response exists and is unique. Observe that each dominating path computed in Algorithm \ref{proc:id} is an earliest-arrival best response. It can also be seen that, when all players take the earliest-arrival best response to each other in $\bfm \alpha$, then $\bfm \alpha$ is exactly the NE computed in Algorithm~\ref{proc:id}.}

{We find that the earliest-arrival best {responses}   can be computed via a dynamic programming that resembles the classical Dijkstra's algorithm for the shortest-path problem. This algorithm is polynomial  when $\Delta(c_r)$ is finite. Note that there may well be best-responses that are not earliest-arrival. However, the earliest-arrival best response possesses several nice properties that are not owned by general best responses.}


{Under the path profile $\bfm\alpha$,  for any edge $e=uv$, recall that $Q_{e}^t$ is the set of players queuing on edge $e$ at time $t\ge r$. For an edge $e'=wu$ whose ending 
{vertex} $u$ is the starting 
{one} of $e$, we use  $$Q_e^t(t,e')\subseteq Q_e^t$$ to denote the subset of players in $Q_e^t$ who entered $e=uv$ at time $t$ and from edges with priorities no higher than $e'=wu$. For any vertex $v\in \bar V$ with $\tau^v\neq 0$, denote $e^*(v)$ as the edge  with the highest priority w.r.t  vertex $v$ that  $\zeta$ can use to 
{reach} $v$ at time $\tau^v$. Initially, we set $Q_e^r(r,e^*(o_\zeta(c_r)))$ as the set of players in $Q_e^r$ that queue after $\zeta$ when $e=e_\zeta(c_r)$, where $e_\zeta(c_r)$ is the initial edge that $\zeta$ rests on in $c_r$.}


\begin{lemma}\label{lem:dp}
For any vertex $v\in \bar V$, if $\tau^v\neq\infty$, then
$$
\tau^v=\min_{uv\in \bar E^-(v)}\Big\{\tau^u+1+|Q_{uv}^{\tau^u}|-|Q_{uv}^{\tau^u}(\tau^u,e^*(u))|\Big\}
$$
and
$$
e^*(v)=\arg\min_{uv\in \bar E^-(v)}\!\!\!{}^*\Big\{\tau^u+1+|Q_{uv}^{\tau^u}|-|Q_{uv}^{\tau^u}(\tau^u,e^*(u))|\Big\},
$$
where $\arg\min^*$ selects the edge with the highest priority among all those achieving the minimum.
\end{lemma}

{Suppose now $\Delta(c_r)$ is finite. Then, the values of $Q_{e}^t$ and $Q_e^t(t,e')$ are polynomially computable through simulating the traffic process. Therefore, the earliest-arrival best response of each player is polynomially computable in this case.}

\subsection{NE properties}

The following lemma is very powerful. It helps us to show that the interactions between players
of different batches at an NE are relatively clear and hierarchal, bearing many similarities
to those in the special SPE studied in Section~\ref{sec:existence} and Appendix~\ref{apx:propspe}.

\begin{lemma}\label{NEproperty'}
Let $\bfm{\pi}$ be an NE of $\Gamma^n(c_r)$. For any batch index $k\ge1$,  player {$i\in B:= \Delta(\bfm\pi,[k])$},
vertex  {$v\in \Pi_i$}, player $j\in\Delta(c_r)\sm B$, and partial path profile {$\bfm r$} for players in $\Delta(c_r)\sm B$,
it holds that
\begin{equation}\label{NEproperty:eq1}
t^v_i(\bfm{\pi})=t^v_i(\bfm{\pi}_{ B}, \bfm r)\le t^v_j(\bfm{\pi}_{ B}, \bfm r) \mbox{ and }
t_j^d(\bfm{\pi}_{B}, \bfm r)\ge \tau(\bfm{\pi},k+1)>t_i^d(\bfm{\pi}_{B}, \bfm r).
\end{equation}
\end{lemma}

One of our main results in this section is the following theorem. {Recall that a strong NE is an NE such that no subset of players are able to strictly better off via group deviation.}

\begin{theorem}\label{cor:nobetter'}
Let $\bfm{\pi}$ be an NE of $\Gamma^n(c_r)$. {The following  properties are satisfied}.
\begin{mylabel6}
\item {\em Hierarchal Independence.} If players in {a batch and those in earlier batches all} follow their equilibrium
strategies as in $\bfm\pi$, {then their arrival times at any vertex are independent of other players' strategies.}

\item {\em Hierarchal Optimality.} The {latency} of each player in {the first batch  $\Delta(\bfm{\pi},1)$}   
 is the smallest among the {latencies} of all players under any   {routing} of $\Gamma^n(c_r)$. 
In general, for all $k\ge 2$, the {latency} of each player in {the $k$th batch  {$\Delta(\bfm{\pi},k)$}}  is the smallest among the {latencies} of all players in $\Delta(c_r)\sm\Delta(\bfm\pi,[k-1])$ under any {routing of $\Gamma^n(c_r)$} in which 
players {in the first {$k-1$} batches} $\Delta(\bfm{\pi},[k-1])$ follow their routes specified by {$\bfm{\pi}$}. 


\item {\em Global {FIFO.}} 
Under $\bfm\pi$, if 
     {there exists a vertex $v\in \bar{V}$ (particularly the origin $o$) such that player $i$ reaches $v$ earlier than $j$ {does} or they reach $v$ at the same time but $i$ comes from an edge   with a higher priority than $j$ does,} 
    then~$i$ reaches the destination $d$ no later than $j$ does.
\item {\em Strong NE.} In $\Gamma^n(c_r)$, every NE is also a {strong NE.} 
That is, in any NE of $\Gamma^n(c_r)$, no set of players can be strictly better off by deviating together to other paths. In particular, this means that each NE of $\Gamma^n(c_r)$ is weakly Pareto Optimal.
    \end{mylabel6}
\end{theorem}

\subsection{Equilibrium relationship between $\Gamma$ and $\Gamma^n$}
\label{sec:relation}

In this subsection, {we establish in a constructive way that each NE outcome of $\Gamma^n$ is an SPE outcome
of $\Gamma$.} Given any NE $\bfm\pi$ of $\Gamma^n$, we construct {\em for {every} history}
$h_r=(c_0,\ldots,c_r)$, {$r\ge0$,} an NE $\bfm\pi(h_r)$ of interim game {$\Gamma^n(c_r)$} with the {\em particular setting
$\bfm\pi(h_0):=\bfm\pi$}. Then, we construct an SPE of $\Gamma$ by assembling {these} NEs such that starting
from any history $h_r$ ($r\ge0)$ the outcome of the SPE is exactly the NE $\bfm\pi(h_r)$. (Note that the reference of
each NE constructed is {a} history {rather than a} configuration. Since different histories may
have the same ending configuration, we may construct multiple NEs 
{for} the same interim game.)

Such an {NE-based} assembling is more difficult than the one in Section~\ref{sec:existence}{, which}
constructs {a} special SPE {restricted to nothing}.
What is more complicated here is that we are unable to design a
{Markovian SPE (cf.~Appendix~\ref{apx:propspe})}. In particular, the natural idea {of} constructing
the NEs $\bfm\pi(h_r)$, $r\ge1$ {rather} than $\bfm\pi(h_0)=\bfm\pi$
directly using Algorithm~\ref{proc:id} does not work, because, e.g., players not in the first batch under $\bfm\pi$
may have incentives to deviate at the game tree root.

Our construction of the NEs is done iteratively on the game tree of $\Gamma$ from the root $h_0=(c_0)$. Initially, the constructed NE $\bfm\pi(h_0)$ for $h_0$ is simply $\bfm\pi$. For each $r\geq 1$, suppose inductively that for a history $h_{r-1}=(c_0,\ldots,c_{r-1})\in\mathcal H_{r-1}$, the NE $\bfm\pi(h_{r-1})$ of game $\Gamma^n(c_{r-1})$, written for convenience as $\bfm\rho= (P_i)_{i\in \Delta(c_{r-1})}$, has  been constructed. We construct below
the NE $\bfm\pi(h_r)$}, denoted $\bfm\rho'=(P'_i)_{i\in\Delta(c_{r})}$,  for each  child history $h_{r}=(c_0,\ldots,c_{r-1},c_r)$ of $h_{r-1}$. Suppose {that}
\begin{mylabel5}
\item[-] $(e_i)_{i\in\Delta(c_{r-1})}$ is the action profile at game tree node $h_{r-1}$ determined by $\bfm\rho$,
 i.e,  no action in the profile  deviates from $\bfm\rho$, {and}
\item[-]  $(e'_i)_{i\in\Delta(c_{r-1})}$ is the action profile {that} leads history $h_{r-1}$  to {its}
 child history $h_{r}$ (or {equivalently} leads $c_{r-1}$ to $c_r$).
\end{mylabel5}
We construct $\bfm\rho'$ in two steps as follows.

\medskip\noindent{\bf CONSTRUCTION I:}
Let $\Bbbk\ge0$ be the {maximum} nonnegative integer 
{such} that the action of each player of $\Delta(\bfm\rho,[\Bbbk])$ under {$(e'_i)_{i\in\Delta(c_{r-1})}$} is the same as 
{under} {$(e_i)_{i\in\Delta(c_{r-1})}$}, i.e., 
\begin{equation}\label{eq:defk}
\Bbbk:=\sup\{k\,|\,e_i=e_i' \text{ for all } i\in\Delta(\bfm\rho,[k])\}.
\end{equation}
({Note that it is possible $\Bbbk=0$ with $\Delta(\bfm\rho,[0])=\emptyset$ or $\Bbbk=\infty$ with
$\Delta(\bfm\rho,[\infty])=\Delta(c_{r-1})$}.) {We let players in the first $\Bbbk$ batches  who are still in the system
at time $r$ under $\bfm\rho$ keep their paths under $\bfm\rho'$. To be more specific, we set}
\begin{equation}\label{eq:construct1}
 P_i':=P_i[o_i(c_r),d]\text{ for all }i\in\Delta(\bfm\rho,[\Bbbk])\cap\Delta(c_r).
\end{equation}

\medskip\noindent{\bf CONSTRUCTION II:}
{Based on}  {the equilibrium strategies kept for players in $\Delta(\bfm\rho,[\Bbbk])\cap\Delta(c_r) $ as specified} in (\ref{eq:construct1}), {which particularly guarantees invariant arrival times at any vertex for these players regardless of other players' choices (see the {\em hierarchal independence} in {Theorem} \ref{cor:nobetter'}(i))}, we 
find  an {\em iterative dominating path  profile}  $(P'_i)_{i\in \Delta(c_r)\setminus \Delta(\bfm\rho,[\Bbbk])}$  for {the remaining players,} 
using a process that is more general than Algorithm \ref{proc:id} {(see Appendices  \ref{apx:general} and \ref{apx:construct} for more details)}.
\medskip

Intuitively,  $\bfm\rho'$ is  {a combination of a part of $\bfm\rho$ and} an iterative dominating partial path profile. It can be shown that the constructed $\bfm\rho'$ is indeed an NE of $\Gamma^n(c_r)$ (see Lemma~\ref{lem:inductive} {in the Appendix}), which completes our {\em inductive} construction. {Furthermore,} the partial hierarchal independence and iterative domination guaranteed by Constructions I and II enable {us} to accomplish our task of assembling all the NEs constructed into an SPE of $\Gamma$.

\begin{theorem}\label{thm:equilibrium-relation}
If $\bfm \pi$  is  an NE of game $\Gamma^n$, then there exists an SPE $\sigma$ of  game $\Gamma$, such that the path profile induced by the initial history $h_0$ and $\sigma$ is exactly $\bfm \pi$.
\end{theorem}

Combining with Example~\ref{eg:SPEnotNE}, the above theorem shows
that the NE outcome set of $\Gamma^n$ is typically a proper subset of the SPE outcome set of $\Gamma$, reaffirming an intuition that
model $\Gamma$ is more flexible than $\Gamma^n$. Since model $\Gamma^n$ is relatively easier to study, also natural
and more frequently analyzed in the literature, Theorem~\ref{thm:equilibrium-relation} can serve as a useful
bridge between $\Gamma$ and $\Gamma^n$.

\section{Bounding equilibrium queue lengths}
\label{sec:bounded-queue}

{In this section, we consider two classes of networks, and prove that the queue lengths
at any NE of the game $\Gamma^n(G)$ are bounded above by a finite number that \emph{only}
depends on the number $m=|E|$ of edges in $G$, provided that the inflow size never exceeds the network capacity.}

\begin{definition}\label{def:sp}
A  network $G$ with origin $o$ and destination $d$ is {\em $o$-$d$ series-parallel} or simply {\em series-parallel} if
\begin{mylabel}
\item {$G$ consists of a single edge {$od$};} or
\item $G$ is obtained by connecting two smaller $o_i$-$d_i$ series-parallel networks $G_i$, $i=1,2$, {\em in series} ---
 identifying $d_1$ and $o_2$, and naming $o_1$ as $o$, and $d_2$ as $d$;  or
\item $G$ is obtained by connecting two smaller $o_i$-$d_i$ series-parallel networks $G_i$, $i=1,2$, {\em in parallel} ---
 identifying $o_1$ and $o_2$ to form $o$ and identifying $d_1$ and $d_2$ to form $d$.
\end{mylabel}
\end{definition}

We reserve $L$ for the length of a longest $o$-$d$ path in $G$, and $\Lambda$ for the maximum in-degree of vertices
in $G$. Then clearly, $\max\{L,\Lambda\}\le m$. The following observation is trivial but important.

\begin{observation}\label{ob:easy}
Under any routing, at most $\Lambda$ players can reach the same vertex (in particular, {the destination} $d$) at the same time.
\end{observation}

In view of Definition~\ref{def:sp}, the series-parallel network $G$ is obtained from $m$ edges by performing a
{sequence} of $m-1$ series or parallel {connection} operations. Each of these operations
connects two series-parallel subnetworks $G_1$ and $G_2$ of $G$ into a bigger series-parallel subnetwork $G_3$
of $G$. {Fix any such sequence of $m-1$ connection operations that leads to $G$ and let} $\mathfrak S$
be the set of all the {subnetworks} $G_1$, $G_2$ {and} $G_3$ {that appear
during the whole process of the sequence of $m-1$ connection operations.} Then clearly,
\begin{center}
$G\in\mathfrak S$ and $|\mathfrak S|=2m-1$.
\end{center}

Henceforth, we study an arbitrary NE, denoted as {$\bfm\pi=(\Pi_i)_{i\in\Delta}$}, of game {$\Gamma^n$}.
For each {player} $i\in\Delta$ and each vertex $v\in \Pi_i$, we use
$\Delta^v_i:=\{j\in\Delta\,|\,t^v_j(\bfm\pi)=t^v_i(\bfm\pi)\}$ to denote the set of players in $\Delta$ who reach $v$
under $\bfm \pi$ at the same time as $i$ does. Note from Observation~\ref{ob:easy} that
\[
|\Delta^v_i|\le\Lambda\text{ for every }i\in \Delta\text{ and each vertex }v\in \Pi_i.
\]
 Using the observation and global FIFO property in {Theorem} \ref{cor:nobetter'}(iii), we can upper bound the ratio between the {population sizes} in any pair of subnetworks connecting in parallel, as the following lemma states.
\begin{lemma}\label{lem:degree}
Suppose that $G_1,G_2\in\mathfrak S$ are connected in parallel to form a series-parallel
network in $\mathfrak S$. Under $\bfm \pi$ and at any time point $t$, if
$G_1$ and $G_2$ accommodate $n_1$ and $n_2$ players respectively, then $n_i\le2\Lambda
L(2\Lambda +n_{j})\le2m^2(2m+n_{j})$ for $\{i,j\}=\{1,2\}$.
\end{lemma}

By {an $o$-$d$ {\em cut} or simply} a {\em cut} of $G$ we mean a set of  edges whose removal from $G$ leaves the
graph unconnected {from $o$ to $d$}. We say that a minimal cut of $G$ is {\em full} at time $t$ if each edge of the cut
has some player on it at time $t$ {(i.e., there is a nonempty queue at each edge). Let $C=\{e_1,\ldots,e_k\}$ and
$C'=\{e'_1,\ldots,e'_k\}$ be two minimum cuts of $G$. We say that $C$ is {\em on the left} of $C'$ if for each $i\in[k]$
there exists an $o$-$d$ path in $G$ which contains $e_i,e'_i$ and visits $e_i$ before $e'_i$.}

{For any $o'$-$d'$ series-parallel network $G'$,} let {$\Xi(G')$} denote the leftmost minimum {$o'$-$d'$} cut of $G'$.
In particular, $|\Xi(G')|$ is the minimum cut size of $G'$. The cut $\Xi(G')$ {partitions} $G'$ into two
parts: the {left part $G'^{l}$ that contains $o'$, and the right $G'^{r}$ that contains} $d'$.

\begin{definition}
Let $\mathfrak F$ denote the set of series-parallel subnetworks $  G_i\in\mathfrak S$
such that there exist  finite integers $F^l_i$ and $F^r_i$ satisfying  {the following two conditions (for any inflow):}
\begin{mylabel}
\item $\Xi(G_i)$ is full as long as $G_i^l$ accommodates more than $F^l_i$ players,\footnote{It means that the full
cut and the accommodation are observed at the same time.} and
\item $G_i^r$ can accommodate at most $F^r_i$ players at any time.
\end{mylabel}
\end{definition}

{First, we can see the set $\mathfrak F$ is {nonempty because} every single edge apparently belongs to $\mathfrak F$.
In the two lemmas below, we are given  $G_1,G_2\in\mathfrak F$, such that $G_i$ is an $o_i$-$d_i$ series-parallel
network  for $i=1,2$.}

\begin{lemma}\label{lem:series1}
If $G_s\in\mathfrak S$ is the combination of $G_1$ and $G_2$ in series, then $G_s\in\mathfrak F$.
\end{lemma}

\begin{lemma}\label{lem:parallel}
{If  $G_p\in\mathfrak S$ is the combination of $G_1$ and $G_2$ in parallel, then $G_p\in\mathfrak F$.}
\end{lemma}
\begin{proof}
{It is clear that $\Xi(G_p)=\Xi(G_1)\cup\Xi(G_2)$, $G^l_p=G^l_1\cup G_2^l$, and $G^r_p=G^r_1\cup G_2^r$. Therefore, $G^r_p$ accommodates at most $F_p^r\equiv F^r_1+F^r_2$ players at any time.}
To prove $G_p\in\mathfrak F$, we only need to show that $\Xi(G_1)\cup\Xi(G_2)$ is full as long as $G^l_1\cup G_2^l$ accommodates more than a certain finite number  of players. As $G_1,G_2\in\mathfrak F$, it suffices to consider the time $t$ when one of $G^l_1$ and $ G_2^l$, say $G^l_1$, accommodates at most $F^l_1$ players. Suppose  that $  G_2^l$ accommodates $F$ players at time~$t$. It follows  from Lemma \ref{lem:degree} that $G_1$  accommodates at least $F/(2m^2)-2m$ players at time $t$. As $G_1\in\mathfrak F$, there are at least $F/(2m^2)-2m-F^r_1$ players inside $G_1^l$ at time $t$. It follows from $F/(2m^2)-2m-F^r_1\le F^l_1$ that $F\le 2m^2(2m+F^l_1+F^r_1)$. \end{proof}

\begin{theorem}\label{thm:queue-length-sp-network}
Let $G$ be a series-parallel network. If $|\Delta_t|\le|\Xi(G)|$ for all $t\ge 1$, then there exists a finite number
such that, for any NE of $\Gamma^n(G)$, the number of players in $G$ is upper bounded by this number, implying that the
latency of any player {at any} NE is upper bounded.
\end{theorem}

\begin{proof}
Since $|\mathfrak S|=2m-1$, combining Lemmas~\ref{lem:series1} and \ref{lem:parallel}, an inductive argument shows that $\mathfrak S=\mathfrak F$. In particular, $G\in\mathfrak F$ says that $\Xi(G)$ is full as long as $G^l$ accommodates more than $F^l$ players,  and $G^r$ can accommodate at most $F^r$ players at any time, where $F^l$ and $F^r$ are finite numbers. Similar to the argument used in the proof of Lemma~\ref{lem:series1}, we  consider  any time point $t$ such that $G^l$ accommodates at most $F^l$ players at time $t$, and more than $F^l$ players at time $t+1$. Then $G^l$ accommodates at most $F^l+|\Delta_{t+1}|\le F^l+|\Xi(G)|$ players and $\Xi(G)$ is full at time $t+1$. Thus  $|\Xi(G)|$ players   leave $G^l$ at time $t+2$ while   $|\Delta_{t+2}|$ ($\le|\Xi(G)|$) players enter $G^l$. It follows  that the number of players inside $G^l$ is nonincreasing unless the number decreases below $F^l+1$. Therefore, at any time $G^l$  can accommodate at most $F^l+|\Xi(G)|$ players, and $G$ can accommodate at most $F^l+|\Xi(G)|+F^r$ players.  {This also means the total latency for any player {travelling} from $o$ to $d$ under any NE is bounded by a finite number since any queue length is upper bounded by the finite number $F^l+m+F^r$.}
\end{proof}

{If we take the  average traveling time from $o$ to $d$ for all players  as a measure of the social welfare as in \cite{sst16}, then the boundedness of the queues implies the PoA of game $\Gamma^n(G)$ is also bounded. 
By Theorem \ref{thm:equilibrium-relation}, this means that the PoS of game $\Gamma(G)$ is also bounded.} 

\medskip {In} closing this section, {we establish boundedness of any {\em SPE queue lengths} for another type of networks.}

\begin{theorem}\label{thm:queue-length-SPE}
Given a dynamic routing game $\Gamma$ {on} an acyclic network $G=(V,E)$ with origin $o$ and destination
$d$, suppose $|E^+(v)|\ge |E^-(v)|$ for any vertex $v\in V\sm\{o,d\}$. Let $\sigma$ be {any} 
SPE of game
$\Gamma(G)$. If the inflow size {$|\Delta_r|$ never exceeds} the size of a minimum $o$-$d$ cut of $G$,  then
there exists a finite number $U$ depending only on $|E|$ such that $t_{i}(\sigma)\le U$ all $ i\in \Delta$.
\end{theorem}

\section{Concluding remarks}
\label{sec:conclusion}
In this paper, we have studied an atomic network congestion game of discrete-time dynamic traffic.
This is a relatively unexplored area in the study of congestion games. The most prominent feature
of our model is the great flexibility agents enjoy so as to make online decisions at all intermediary
vertices and, accordingly, SPE serves as the default solution concept. We have shown that this more
flexible model has close connections with the corresponding game of normal form, which is 
{more often}
studied in the literature. We have identified many surprisingly nice properties of the NE flows of
the latter model, the equivalence between NEs and strong NEs and a global FIFO, to name a few.

This paper is our first attempt in understanding the consequences of the introduction of agents'
flexibility of online decision making in dynamic traffic games. Many interesting problems are widely open.
For example, {is there} an upper bound on the SPE (or NE) queue lengths for general networks? What
are more accurate bounds on PoA and PoS  w.r.t.\ either NEs or SPEs? {Does a long-run steady
state exist} when the inflow is 
{constant} or seasonal? How efficient is this steady state if it does exist?
What if agents have multi-origins and multi-destinations? What if agents are allowed to choose their departure
times? Exploring these problems will undoubtedly help us better understand atomic games of dynamic traffic.

\bibliography{traffic}

\medskip
\begin{appendix}
\begin{center}
{\bf \Large {Appendix}}
\end{center}

\section{Details in Section 3}

\subsection{Iterative dominations}
To facilitate our discussions, the proof of Lemma \ref{lem:ineq} in particular, we introduce several notations. Given any (directed) path $P$ in $\bar{G}$, and vertices $a,b$ in $P$ such that $a$ is passed no later than $b$ by path
$P$, we use $P[a,b]$ to denote the sub-path of $P$ from $a$ to $b$. We write $P(a,b]=P[a,b]\sm\{a\}$,
$P[a,b)=P[a,b]\sm \{b\}$ and $P(a,b)=P[a,b]\sm \{a,b\}$.

Given any feasible configuration $c_r$ at time $r$, let   players $1,2,\ldots$ of $\Delta(c_r)$ be as  indexed and path profile $\bfm\pi=(\Pi_i)_{i\in\Delta(c_r)}$ be as computed in Algorithm \ref{proc:id}. For any player indices $i,j$ with $i<j$ and any vertex   $v\in\bar G$, let 
{\begin{equation}\label{eq:deftau}
\tau_j^v\llbracket i\rrbracket:=\min\{t_j^v(\bfm\pi_{[i]},R_j)\,|\,R_j\in\mathcal P[o_j(c_r),d]\}\end{equation}}
denote the value {$\tau^v_j$  computed for player $j$} in Step 3 at the $(i+1)$th iteration of Algorithm \ref{proc:id}, {i.e.,} the earliest time for player
$j$ to reach vertex $v$ starting from edge $e_j(c_r)$, based only on the partial routing $\bfm\pi_{[i]}$ of players in $[i]$.

\begin{lemma}[Restatement of Lemma \ref{lem:ineq}]
Let player indices $i,j$ and player subset $S$ satisfy $j\in S\subseteq\Delta(c_r)\sm[i-1]$. Then
\begin{equation}\label{eq:conclusion}
t_i^v(\bfm\pi_{[i]}, \bfm p_{S\,\sm\,\{i\}})=\tau_i^v\llbracket i-1\rrbracket \le t_j^v(\bfm\pi_{[i-1]},
\bfm p_S)\end{equation} holds for every vertex $v\in\Pi_i$ and every path profile $\bfm p$ of $\Gamma^n(c_r)$.
\end{lemma}

 \begin{proof}We prove  by induction on $i$. Consider first the base case $i=1$, whose proof is quite similar to the general case. First, $t_1^v(\bfm\pi_{[1]}, \bfm p_{S\,\sm\,\{1\}})\ge \tau_1^v\llbracket 0\rrbracket$ is apparent, because $\tau_1^v\llbracket 0\rrbracket$ is the shortest path length without waiting cost from $o_1(c_r)$ to $v$.
To see that $t_1^v(\bfm\pi_{[1]}, \bfm p_{S\,\sm\,\{1\}})=\tau_1^v\llbracket 0\rrbracket$, it suffices to prove that player 1 never queues on any edge under routing $(\bfm\pi_{[1]}, \bfm p_{S\,\sm\,\{1\}})$. Suppose on the contrary that player 1 queues after some player $k\in S$ on edge $xy$, and let $xy$ be the first of such edge on $\Pi_1$. Then, under  $(\bfm\pi_{[1]}, \bfm p_{S\,\sm\,\{1\}})$, either player $k$ enters $xy$ earlier than {player 1} or they enter $xy$ at the same time but player $k$ comes from an edge with a higher priority.  Define a path of $k$ as $X_k:=P_k[o_k(c_r),y]\cup \Pi_1[y,d]\in \mathcal P[o_k(c_r),d]$. Then, {$t_k^w(X_k)\leq t_1^w(\Pi_1)$ for all vertices $w\in \Pi_1[y,d]$. {Note that $\tau_1^w\llbracket 0\rrbracket=t_1^w(\bfm\pi_{[1]}, \bfm p_{S\,\sm\,\{1\}})$ and  $\tau_k^w\llbracket 0\rrbracket\leq t_k^w(\bfm\pi_{[1]}, \bfm p_{S\,\sm\,\{1\}})$ for all $w\in \Pi_1[y,d]$. Then, either (i) $\tau_1^x\llbracket 0\rrbracket>\tau_k^x\llbracket 0\rrbracket$, or (ii) $\tau_1^x\llbracket 0\rrbracket=\tau_k^x\llbracket 0\rrbracket$ and via $X_k$ player $k$ is able to enter $xy$ from an edge with a higher priority and follow the remaining path of $\Pi_1$,}} or (iii) $e_1(c_r)=e_k(c_r)=xy$ and player 1 queues after $k$ on this edge, all contradicting the choice of player 1.
This finishes the first part of (\ref{eq:conclusion}) for $i=1$. The second part of (\ref{eq:conclusion}) for $i=1$ is obvious, whether $j=i$ or not.

Suppose now $i\ge 2$ and (\ref{eq:conclusion}) is valid for smaller values. {This means that, as long as {players in $[i-1]$ follow $\bfm\pi_{[i-1]}$, they} will never be affected by other players; in fact, no player $j\ge i$ can reach any vertex earlier than any one in $[i-1]$.}  This also means that, as long as  {players in $[i-1]$ follow $\bfm\pi_{[i-1]}$,  they will} exert {\em invariant influences} on the movements of other players: the set $Q_e^s\cap [i-1]$ depends only on time $s\ge r$ and edge $e$ but not on the choices of players in $S\sm [i-1]$; in addition, players in $Q_e^s\cap [i-1]$  {(if nonempty)} always queue before other players. It is {this} invariant influence property that makes the proof of the general case almost the same as the base case, as demonstrated below.

{First, due to the above invariant influences, it's not hard to see that $t_i^v(\bfm\pi_{[i]}, \bfm p_{S\,\sm\,\{i\}})\geq \tau_i^v\llbracket i-1\rrbracket$: players in $S\sm [i]$ cannot speed up $i$, because the latency caused by players in $[i-1]$ alone is always $\tau_i^v\llbracket i-1\rrbracket$, regardless of the choices of  $S\sm [i]$.
Combining $\tau_i^v\llbracket i-1\rrbracket\le\tau_j^v\llbracket i-1\rrbracket$ {(because otherwise, using $X_j:=P_j[o_j(c_r),v]\cup \Pi_i[v,d]\in\mathcal P[o_j(c_r),d]$, player $j$ 
{would be} able to reach $v$ earlier than $i$ does and take the remaining part of $\Pi_i$, contradicting the choice of $i$ and $\Pi_i$)} and $\tau_j^v\llbracket i-1\rrbracket \le t_j^v(\bfm\pi_{[i-1]},
\bfm p_S)$ (due to the invariant influences from $[i-1]$) also gives the inequality part of (\ref{eq:conclusion}).}



{So it remains to show that $t_i^v(\bfm\pi_{[i]}, \bfm p_{S\,\sm\,\{i\}})\leq \tau_i^v\llbracket i-1\rrbracket$, i.e., players in $S\,\sm\,\{i\}$ will not slow down $i$.  Suppose 
  {on} the contrary that
$t_i^v(\bfm\pi_{[i]}, \bfm p_{S\,\sm\,\{i\}}) > \tau_i^v\llbracket i-1 \rrbracket=t^v_i(\bfm\pi_{[i]})$
 for some vertex $v\in \Pi_i$. Let $v$ be the first such
vertex along $\Pi_i$,  indicating that
\begin{mylabel}
\item[(1)]  $t_i^{w}(\bfm\pi_{[i]}, \bfm p_{S\,\sm\,\{i\}})
           =\tau_i^{w}\llbracket i-1 \rrbracket$  for every vertex ${w}\in \Pi_i[o_i(c_r),v)$.
\end{mylabel}}

{In view of  the invariant influences from players in $[i-1]$, there must exist some player $k\in S\sm[i]$ and an edge $xy\in \Pi_i[o_i(c_r),v]$ such that $k$ enters $xy$ earlier than $i$ does under  $(\bfm\pi_{[i]}, \bfm p_{S\,\sm\,\{i\}})$. Let $xy$ be the {\em first} such edge along $\Pi_i[o_i(c_r),v]$. By (1), we have
\begin{mylabel}
\item[(2)] under routing $(\bfm\pi_{[i]},\bfm p_{S\,\sm\,\{i\}})$, player $k$
           reaches vertex $x$ and enters edge $xy$ at time $t_k^x(\bfm\pi_{[i]},\bfm p_{S\,\sm\,\{i\}})$ $\le t_i^x(\bfm\pi_{[i]},\bfm p_{S\,\sm\,\{i\}})=\tau_i^x\llbracket i-1 \rrbracket$.
\end{mylabel}
}

{Construct a
path $X_k:=P_k[o_k(c_r),x]\cup \Pi_i[x,d]\in\mathcal P[o_k(c_r),d]$ for $k$. Since  $t^x_k(\bfm\pi_{[i-1]},X_k)= t^x_k(\bfm\pi_{[i-1]},P_k)\le t_k^{x}(\bfm\pi_{[i]},\bfm p_{S\,\sm\,\{i\}})$ (where the equality is simply from definition of $X_k$ and the inequality is due to the invariant influences of $[i-1]$ that $k$ cannot be sped up by players in $S\cup \{i\}$),
which together with (2) implies
$t^x_k(\bfm\pi_{[i-1]},X_k)\le \tau_i^x\llbracket i-1 \rrbracket=t^x_i(\bfm\pi_{[i-1]},\Pi_i)$.
Consequently,
\begin{mylabel}
\item[(3)] $t^{w}_k(\bfm\pi_{[i-1]},X_k)\le \tau_i^{w}\llbracket i-1 \rrbracket$
           for each vertex ${w}\in X_k[x,d]=\Pi_i[x,d]$.
\end{mylabel}}

{By definition of player $i$ from Algorithm~\ref{proc:id} and  (3), we have
$t^{w}_k(\bfm\pi_{[i-1]},X_k)=\tau_i^{w}\llbracket i-1 \rrbracket$ for  each vertex
${w}\in X_k[x,d]=\Pi_i[x,d]$. {Therefore,}
from (3), we know
either (if $X_k$ and $\Pi_i$ have different incoming edges into $x$) $X_k$ has a higher priority incoming edge into $x$ than $\Pi_i$ does, or (by the choice of edge $xy$)
$e_k(c_r)=xy=e_i(c_r)$, $X_k=\Pi_i$, and player $k$ queues before player $i$ 
{at} $e_k(c_r)=e_i(c_r)$. 
However, the choice made at  {the} $i$th iteration 
of Algorithm~\ref{proc:id}
excludes the possibilities of both cases. This finishes the proof.}
\end{proof}

\subsection{Generalized iterative dominations}\label{apx:general}
{As can be seen from the proof of Lemma \ref{lem:ineq},  our induction hypothesis only involves the equation (the first part) of (\ref{eq:conclusion}), which guarantees the critical invariant influences property.} 
This leads us to the following generalization of Algorithm \ref{proc:id}, which computes an iterative dominating {\em partial} path profile {\em based} on the fixed routing of some players.

 \begin{algorithm} \SetAlgoNoLine \KwIn{a feasible configuration $c_r$ at time $r$;   a partial path profile $\bfm b=(B_\ell)_{\ell\in U}$ for players in  a {(possibly empty)} finite subset $U\subseteq\Delta(c_r)$ such that, for all $\ell \in U$ and all $v\in B_\ell$, $t_\ell^v(\bfm b, \bfm p_S)$ is the same  over 
all path profile $\bfm p$ of $\Gamma^n(c_r)$ and $S\subseteq \Delta(c_r)\sm U$.}
\KwOut{the iterative dominating partial path profile (routing) $\bfm\pi=(\Pi_i)_{i\in\Delta(c_r)\,\sm\, U}$ for $\Delta(c_r)\sm U$ along with the corresponding player indices 1,2,\ldots.}
\begin{mylabel5}
\item[\,\,\,1.] 
{Initiate $D\leftarrow U$,
          $\bfm\pi_{[0]}\leftarrow\emptyset$, $i\leftarrow0$.}
\item[\,\,\,2.]\vspace{-1mm} $i\leftarrow i+1$  {(NB: Start to search for the new dominator and his dominating path).}
\item[\,\,\,3.]\vspace{-1mm} 
For each player $j\in \Delta(c_r)\sm D$ and vertex $v\in\bar G$,
\begin{mylabel4}
\item[-] let {$\tau_{j}^v=\min\{t^v_j(\bfm b,\bfm\pi_{[i-1]},R_j)\,|\,R_j\in\mathcal P[o_j(c_r),d]\}$}
          be the earliest time for {$j$} to reach vertex $v$ from his {\em current location in $c_r$}, assuming
          that \emph{all  other} players in {$\bar G$} are those in $D$ and they    {go along their paths specified in $(\bfm b,\bfm\pi_{[i-1]})$};
          \item[-] {let $\mathcal P_{j}^v$ denote} the set of all the \emph{corresponding} {$o_j(c_r)$-$v$ paths for $j$ to reach $j$ at time          $\tau_{i}^v$}. 
          \item[-] if there is no path in $\bar G$ to $v$ from the current location
          of player {$j$}, then set {$\tau_j^v\leftarrow\infty$ and $\mathcal P_j^v\leftarrow\emptyset$}.
          \end{mylabel4}
\item[\,\,\,4.]\vspace{-1mm} {Run Steps 4-- 12 of Algorithm \ref{proc:id} to identify dominator $i$ and his dominating path $\Pi_i$.}
\item[\,\,\,5.]\vspace{-1mm}   {$D\leftarrow D\cup \{i \}$, $\bfm\pi_{[i]}\leftarrow(\bfm\pi_{[i-1]},\Pi_i)$. (NB: the algorithm outputs player $i$ and his path $\Pi_i$.)}
\item[\,\,\,6.]\vspace{-1mm} Go to Step 2.
\end{mylabel5}
\caption{\sc(Iterative Dominating Partial Path Profile with a Base)}\label{alg:gernal}
\end{algorithm}

The verbatim adaption of the proof {of} Lemma \ref{lem:ineq} gives the following generalization for iterative domination, which {will play a critical role in the discussion of {Section \ref{sec:NE-vs-SPE}}.} 
\begin{lemma}\label{lem:basedominant}
Given the input and output of Algorithm \ref{alg:gernal}, if $j\in S\subseteq\Delta(c_r)\sm(U\cup[i-1])$, then for every vertex $v\in \Pi_i$ and  path profile  $\bfm p$ of $\Gamma^n(c_r)$,  it holds that
\[
t_i^v(\bfm b,\bfm\pi_{[i]}, \bfm p_{S\,\sm\,[i]})=\min_{R_i\in\mathcal P[o_i(c_r),d]}t_i^v(\bfm b,\bfm\pi_{[i-1]},R_i)\le t_j^v(\bfm b,\bfm\pi_{[i-1]}, \bfm p _{S}).\]
\end{lemma}

\subsection{Nice properties of the special SPE}\label{apx:propspe}
{\begin{definition}\label{def:preempt}[Preemption] Given a strategy profile $\bfm p$ of interim game $\Gamma^n(c_r)$,  
we say that  $i$ {\em preempts} $j$ at vertex $v$ (under $\bfm p$) if {they both pass $v$ and} $i$ reaches $v$  earlier than $j$ does (under $\bfm p$); we say $i$ {\em weakly preempts} $j$ at vertex $v$ if either $i$ preempts $j$ at vertex $v$, {or}  $i$ and $j$ reach  $v$ at the same time but $i$ comes from an edge {(with head $v$)} with a higher priority than $j$ does.
\end{definition}}
\begin{definition}\label{def:original}
A player $i\in\Delta$ is said to have a \emph{higher original priority} than player $j\in\Delta$ if $i$ preempts or weakly preempts $i$ at the {origin $o$.}
\end{definition}
{It  can be summarized from the above discussions} that the special SPE $\sigma^*$ which we construct in the proof of Theorem \ref{thm:SPE} has the following related but distinct nice properties.

\medskip
{\em{Sequential Independence.} If} a player $i$ and all players 
{with} higher original priorities than $i$
 fix their strategies as in the SPE, then their realized paths as well as the arrival times at
 all vertices are independent of other players' strategies.

{\em Sequential Optimality.} The latency of any player realized under the SPE is minimum among all the feasible flows
 in which players 
 {with} higher original priorities follow their strategies in the SPE.

{\em Pareto Optimality.} No group of players can be strictly better off by deviating together from $\sigma^*$. That is,
 the SPE outcome from any configuration is a strong NE (\cite{a59}) of the corresponding subgame. Thus  the SPE outcome of $\sigma^*$ is weakly Pareto optimal.

{\em Global {FIFO}.} 
In any routing realized under the SPE,  
if  player $i$ 
{weakly preempts} player $j$ at some vertex of $G$ (particularly its origin $o$), then $i$ leaves 
the system {no later than $j$.} 

{\em No Overtaking.} If player $i$ has a higher original priority than player $j$ and both pass through some vertex $v$
 in the realization of the SPE, then $i$ 
 {weakly preempts} $j$ at $v$. 

{\em Earliest Arrival.} Given the other players' strategies in the SPE, each player using his realized path under the SPE is
 guaranteed to reach any vertex on the path (not only the destination $d$) at an earliest time among all of his possible
 {choices of paths.} 

{\em Markov and Anonymity.} According to the SPE, the action each player takes at each node of the game tree of $\Gamma$ depends
 \emph{only} on the 
 {immediately previous configuration but not on earlier configurations in the history,  and the identities of
 other players do not matter}.\footnote{A strategy $\sigma_i$ of player $i$ is called {\em Markovian} if $\sigma_i(h_r)=\sigma_i(h_{r'})$ holds for {all} histories $h_r=(c_0,c_1,\ldots,c_r)$ and $h_{r'}=(c'_0,c'_1,\ldots,c'_{r'})$ with $c'_r=c_{r'}$ and $i\in\Delta(c_r)$.}

\section{Details in Section 4}
\subsection{{Proof of  Lemma \ref{dominating}}}

\begin{proof}
For each vertex $v\in \bar V$, we denote $I_v\subset\Delta(c_r)$ {as} the set of players apart from $\zeta$ whose arrival {times at} $v$ can be influenced by player $\zeta$'s unilateral 
{path} changing.  {Apparently,}  if $j\in I_v$, then it must be the case that $v\in A_j$. {Thus,} we need to prove{, for} every player $j\in I_v$, {that} player $\zeta$ dominates player $j$ at vertex $v$. Suppose edge $(o_\zeta(c_r),v_0)$ is the {starting} edge of path $A_\zeta$. 

\begin{claim}\label{clm:path}
If $I_v\neq\emptyset$, then there exists a directed path from $v_0$ to $v$, i.e., $\tau^v\neq\infty$.
\end{claim}

\begin{proof}
Suppose $j\in I_v$. Let $e_v(A_j)=(u,v)$. Since {there exists} $A'_\zeta\in \mathcal P[o_\zeta(c_r),d]$ such that $t_j^v(A'_\zeta,\bfm{\alpha}_{-\zeta})\neq t_j^v(\bfm{\alpha})$, {one of the following cases must happen}:
\begin{mylabel}
\item $(u,v)\in A'_\zeta$ or $(u,v)\in A_\zeta$;
\item $I_u\neq\emptyset$.
\end{mylabel}

This is {true} because, if $I_u=\emptyset$ and $(u,v)\not\in A'_\zeta\cup A_\zeta$, then $j$'s arrival time {at} $u$ will be a constant and $j$'s 
{waiting} time on edge $(u,v)$ will also be a {constant,} 
contradicting the fact {that} $j\in I_v$. If case (i) happens, then obviously there is a path from $v_0$ to $v$. Otherwise, $I_u\neq\emptyset$,  we 
{can find a directed path from $v_0$ to $v$ by backward induction}.
\end{proof}

Define $V_\zeta:=\{v\in\bar V| \tau^v\neq\infty\}$. 
{It follows from Claim \ref{clm:path} that} $\{v|I_v\neq\emptyset\}\subseteq V_\zeta$. Since the graph $\bar G$ is acyclic, there exists a full order among the vertices in $V_\zeta$, such that each directed edge's tail vertex's order is smaller than the head vertex's order. {Apparently,} $v_0$'s order is the smallest.  So it {suffices} to prove that{, for} any vertex $v\in V_\zeta$, player $\zeta$ dominates every player $j\in I_v$ at vertex $v$. We will prove this  by induction on the order of the vertices in $V_\zeta$. The base case is {obvious because} $I_{v_0}=\emptyset$. Now  for some vertex $w\in V_\zeta$, assume the above statement is true for all vertices  with {orders smaller} than $w$ in $V_\zeta$, we prove that it's also 
{the case} for vertex $w$.

{Since the case $I_w=\emptyset$ is trivial, we suppose now} $I_w\neq\emptyset$. For any player $j\in I_w$, suppose edge $(u,w)\in A_j$. In the following, we {prove first that} player $\zeta$ dominates player $j$ at vertex $u$, then {show} the dominance at vertex $w$. If $j\in I_u$, since $u$ has a smaller order than $w$, then by the assumption, player $\zeta$ dominates player $j$ at vertex $u$. If $j\not\in I_u$, then no matter how $\zeta$ {changes his path}, player $j$'s arrival time {at} $u$ is a constant. However, since $j\in I_w$, {there exists a path} $A'_\zeta\in \mathcal P[o_\zeta(c_r),d]$ such that $t_j^w(A',\bfm{\alpha}_{-\zeta})\neq t_j^w(\bfm{\alpha})$. Suppose 
{w.l.o.g.} that $t_j^w(A',\bfm{\alpha}_{-\zeta})<t_j^w(\bfm{\alpha})$.  {Then,} combining the above two facts of $j\not\in I_u$ and $j\in I_w$, we know  one of the  {two following} cases must {happen}:
\begin{mylabel}
\item $\exists i\in I_u$ with $uw\in A_i\cap A_j$,  s.t.,  $t_i^u(\bfm{\alpha})<t_j^u(\bfm{\alpha})$,   or $t_i^u(\bfm{\alpha})=t_j^u(\bfm{\alpha})$ {and} $e_u(A_i)\prec_u e_u(A_j)$.
\item $uw\in A_\zeta\cap A_j$, and $t^u_\zeta(\bfm{\alpha})<t_j^u(\bfm{\alpha})$, or $t^u_\zeta(\bfm{\alpha})=t_j^u(\bfm{\alpha})$ {and} $e_u(A_\zeta)\prec_u e_u(A_j)$.
\end{mylabel}
If it is the case (i), then combining the assumption that $\zeta$ {dominates} all players in $I_u$ at vertex $u$ and the fact {that} $\tau_j^u=t_j^u(\bfm{\alpha})$, we can deduce that $\zeta$ dominates player $j$ at vertex $u$; If it is the case (ii), then apparently $\zeta$ still dominates player $j$ at vertex $u$. Next we prove $\zeta$ dominates $j$ at vertex $w$.

{It can be observed from the above analysis that} no matter whether 
{$j\in I_u$ or not}, player $\zeta$ always dominates player $j$ at vertex $u$. Now suppose path $A^*_\zeta$  satisfies $\tau^u=t^u_\zeta(A^*_\zeta,\bfm{\alpha}_{-\zeta})$ (if there are more than one such {paths, then} let $A^*_\zeta$ be the one with {the} highest edge priority at vertex $u$),  and $ A'_\zeta\in\mathcal P(o_\zeta(c_r),d)$ satisfies $\tau^w_j=t^w_j(A'_\zeta,\bfm{\alpha}_{-\zeta})$. Define $\bar A_\zeta=A^*_\zeta[o_\zeta(c_r),u]\cup A_j[u,d]$, then apparently $\bar A_\zeta\in \mathcal P[o_\zeta(c_r),d]$. Under the strategy profile $(\bar A_\zeta, \bfm{\alpha}_{-\zeta})$, consider first {the case that} $\zeta$ travels along the edge $uw$ and there is no queue. 
{In this case,} $t^w_\zeta(\bar A_\zeta, \bfm{\alpha}_{-\zeta})=\tau^u+1\le \tau^u_j+1\le t^u_j(A'_\zeta,\bfm{\alpha}_{-\zeta})+1\le t^w_j(A'_\zeta,\bfm{\alpha}_{-\zeta})=\tau^w_j$. Combining with the facts that $\tau^w\le t^w_\zeta(\bar A_\zeta, \bfm{\alpha}_{-\zeta})$ and $e_w(\bar A_\zeta)=e_w(A_j)$, we can deduce that player $\zeta$ dominates player $j$ at vertex $w$. Now we are 
{left} with the case {that} $\zeta$ travels along the edge $uw$ under the strategy profile $(\bar A_\zeta, \bfm{\alpha}_{-\zeta})$ {and} there is a queue. Let $I'$ be the set of players 
{queuing before} $\zeta$ 
and 
{those who} pass through $uw$ earlier than that queue. {Let} $i\in I'$ {be} the 
{player that queues immediately before $\zeta$,} 
i.e., $t^w_\zeta(\bar A_\zeta, \bfm{\alpha}_{-\zeta})=t^w_i(\bar A_\zeta, \bfm{\alpha}_{-\zeta})+1$. Since $\zeta$ dominates all players in $I_u$ and $t^u_\zeta(\bar A_\zeta, \bfm{\alpha}_{-\zeta})=\tau^u$,  we can see that $I'\cap I_u=\emptyset$ and $\zeta$ {cannot} dominate any player in $I'$ at vertex $u$. So no matter how $\zeta$ changes his strategy, the time that players in $I'$  pass through 
{edge} $uw$ will never be influenced by $\zeta$ 
{or} players in $I_u$, which also means $I'\cap I_w=\emptyset$. {Thus,} $t^w_i(A'_\zeta, \bfm{\alpha}_{-\zeta})=\tau^w_i=t^w_i(\bar A_\zeta, \bfm{\alpha}_{-\zeta})$. Recall that $\zeta$ dominates $j$ at vertex $u$ and $uw\in A_j\cap A_i$. {Therefore,} no matter how $\zeta$ {chooses} his path strategy, player $j$ will always arrive at vertex $w$ at least one {unit of} time later than $i$ does. {So,} by the definition of path $A'_\zeta${, we} have $\tau_j^w=t^w_j(A'_\zeta,\bfm{\alpha}_{-\zeta})\ge t^w_i(A'_\zeta,\bfm{\alpha}_{-\zeta})+1=t^w_\zeta(\bar A_\zeta, \bfm{\alpha}_{-\zeta})$. This along with the facts {that} $\tau^w\le t^w_\zeta(\bar A_\zeta, \bfm{\alpha}_{-\zeta})$ and $e_w(\bar A_\zeta)=e_w(A_j)$ still implies that player $\zeta$ dominates player $j$ at vertex $w$. By the arbitrariness of player $j$ in $I_w$, we can see that player $\zeta$ dominates all players in $I_w$ at vertex $w$.
\end{proof}

\subsection{Proof of Lemma \ref{lem:dp}}

\begin{proof}

We only need to show the first equality because the second one is  true by definition  {and the correctness of the first}. Since the network is acyclic, we have a natural full order among all vertices that player $\zeta$ can reach, i.e., the vertices  {$v$} with $\tau^v\neq\infty$. We prove by 
{induction} on the order of these vertices. The base case that $v$ is {the head of} the initial edge $e_\zeta(c_r)$ is obvious due to the initial setting. Let  {us} now consider the case that $v$ is not the  {head} of $e_\zeta(c_r)$, and suppose the lemma is true for all vertices with  orders smaller than $v$.

We claim that,  for the players in $Q_{uv}^{\tau^u}-Q_{uv}^{\tau^u}(\tau^u,e^*(u))$, no matter how $\zeta$ chooses his path, their arrival times at $u$ will never be influenced. Suppose the contrary. Then,  {by} Lemma~\ref{dominating}, $\zeta$ dominates at least one player $j\in Q_{uv}^{\tau^u}-Q_{uv}^{\tau^u}(\tau^u,e^*(u))$. Note first from the definition of $Q_{uv}^{\tau^u}$  that $\tau^u_j\leq \tau^u$, where $\tau^u_j$ is the earliest time that $j$ can reach $u$ when $\zeta$ changes 
{his} path. By the definition of domination, it can only be the case that $\tau^u_j=\tau^u$ and $\zeta$ is able to arrive at $u$ at time $\tau^u$ via an edge  {$e'$} that has a priority no lower  than the one taken by $j$.  {By definition of $e^*(u)$, the priority of $e^*(u)$ is at least that of $e'$, and therefore at least that of the edge taken by $j$}. However, this is impossible because $j\notin Q_{uv}^{\tau^u}(\tau^u,e^*(u))$. Hence the claim is valid. It follows from the claim  {and the definition of $\tau^u$} that $\zeta$ cannot influence the arrival times of players in $Q_{uv}^{\tau^u}-Q_{uv}^{\tau^u}(\tau^u,e^*(u))$ at $v$ either.

Consequently,  if $\zeta$ uses edge $uv\in E^-(v)$ to reach $v$,  {his} arrival time  at $v$ is at least $\tau_u+1+|Q_{uv}^{\tau^u}|-|Q_{uv}^{\tau^u}(\tau^u,e^*(u))|$. On the other hand, the above value is obtainable  by reaching $u$ at $\tau^u$ via $e^*(u)$. It follows that the earliest time that $\zeta$ can reach $v$ via edge $uv$ is exactly  $\tau_u+1+|Q_{uv}^{\tau^u}|-|Q_{uv}^{\tau^u}(\tau^u,e^*(u))|$. Since $\zeta$ must use one edge in $E^-(v)$ to reach $v$,  the minimization is correct for vertex $v$.  This finishes the proof.
\end{proof}

\subsection{NE properties of interim games}
{Recall the notions of preemption and weak preemption in Definition \ref{def:preempt}.} If player $i$ weakly preempts player $j$ at a vertex $u$ and  {both} $i$ and $j$ choose to enter the same edge $uv$, then $i$ preempts $j$ at vertex $v$. It is possible that player $i$ preempts player $j$ at vertex $x$ and $j$ preempts $i$ at another vertex $y$ (even under  equilibrium routings).
{Note that while the notion of domination  compares the arrival times of two players at the same node  under  possibly different strategy settings, weak preemption compares two arrival times under the same strategy setting.}



Letting {$\Gamma^n(c_r)$} be an interim game with {$c_r=(Q^r_e)_{e\in\bar E}$} and $S$ a subset of $\Delta(c_r)$, we use {$\Gamma^n(c_r,S)$} to denote the game where only the ones in $S$ play the game (the other players are assumed to disappear and the {orders of players of $S$} 
in any queue $Q_e^t$, {$e\in\bar E$} are modified in accordance). {In contrast to Example \ref{better}, we have the following useful lemma.}


\begin{lemma}\label{lem:preempt}Let $\Gamma^n(c_r)$ be an interim game {and} $S$ a player subset of $\Delta(c_r)$. {Fix} $\bfm f$ as a strategy profile for players in $S$. If, for any player  {$j\in \Delta(c_r)\sm S$}, he can never {weakly} preempt any player in $S$ with any $o_j(c_r)$-$d$ path in game $\Gamma^n(c_r,S\cup\{j\})$, then in game $\Gamma^n(c_r)$, for any strategy profile $\bfm r$ of players in  {$\Delta(c_r)\sm S$}, no player outside $S$ can {weakly} preempt any player of $S$ under $(\bfm f, \bfm r)$.
\end{lemma}
\begin{proof}Suppose on the contrary that
$j\notin S$ {weakly} preempts
$i\in S$ at some vertex $v$ under $(\bfm f, \bfm r)$, and {further that} $t_j^v(\bfm f, \bfm r)$ is the minimum. 
The minimality implies that players in {$\Delta(c_r)\sm (S\cup \{j\})$ do not {weakly} preempt those in $S$ {before} 
$t_j^v(\bfm f, \bfm r)$. So} deleting them can only possibly reduce $j$'s queuing time {before time $t_j^v(\bfm f, \bfm r)$} and accelerate his arrival time at $v$.
Therefore,  $j$ will still {weakly} preempt $i$ at $v$ under {$(\bfm f, \bfm r_{\{j\}})$} 
in game {$\Gamma^n(c_r,S\cup\{j\})$}, {contradicting} the hypothesis. \end{proof}

\begin{proof}[Proof of Lemma \ref{NEproperty'}]
For each player {$j\in \Delta(c_r)\sm B$}, define  {$\tau(B,j)$} as the earliest time when $j$ {can} {weakly} preempt some player of {$B$} in game {$\Gamma^n(c_r,B\cup\{j\})$}  {under strategy profile $(\bfm{\pi}_{B},R_j)$   among all possible $o_j(c_r)$-$d$ paths {$R_j\in\mathcal P[o_j(c_r),d]$}.  If} player $j$ {can} never {weakly} preempt  {any player in} {$B$}  {\em in this sense, we set $\tau(B,j):=\infty$}.
  We show that  $\min\{\tau(B,j)\,|\,j\in  \Delta(c_r)\sm B\}=\infty$, 
{which,} together with Lemma \ref{lem:preempt}, will imply that the  {inequality} 
  of (\ref{NEproperty:eq1}) is true. The  {equality} 
  of (\ref{NEproperty:eq1}) will also be valid because  {as long as players in {$B$} follow $\bfm{\pi}_{B}$, they are not  affected by the remaining ones}.

 {Assume the contrary and} {let $B$ be the smallest nonempty set $\Delta(\bfm\pi,[k])$ such that
 $\tau^*:=\min\{\tau(B,j)\,|\,j\in  \Delta(c_r)\sm B\}<\infty$.} 
This means that,  {\em before time $\tau^*$, no player in {$\Delta(c_r)\sm B$} can {weakly} preempt any one in {$B$} in the sense described above}. {By definition, there exists a player $j\in\Delta(c_r)\sm B$ who {weakly} preempts some player $i^*\in B$  under $(\bfm{\pi}_{B},R_j)$ for some path $R_j\in\mathcal P[o_j(c_r),d]$ at some vertex $u\in R_j\cap \Pi_{i^*}$.} 
{Therefore, under $(\bfm{\pi}_{B},R_j)$,  player $i^*$ just reaches a vertex  $u'\in\Pi_{i^*}[o_{i^*}(c_r),u]$ at time $\tau^*$.}
{Since no weak preemption happens before time $\tau^*$, {the} movement of every player in $B$ before time $\tau^*$ will not be affected. By the minimality of $\tau^*$, we have  $\tau^*=t_{i^*}^{u'}(\bfm{\pi}_{B},\mathbf r)$ {for all} partial path profile $\mathbf r$, {which,} along with the trivial relation $t_{i^*}^{u'}(\bfm{\pi}_B, \mathbf r)\le t_{i^*}^{u}(\bfm{\pi}_B, \mathbf r)$, {gives}} \[\tau^*\le t_{i^*}^u(\bfm{\pi}_B, \mathbf r), \forall \mathbf r.\]
We  define  $j^*$, using  an adaptation of Algorithm \ref{alg:gernal} with vertex  $u$ (resp. $B$, $\bfm\pi_B$) in place of destination $d$ (resp. $U$, $\bfm b$) over there, as the player who owns a  {dominating $o_{j^*}(c_r)$-$u$ path $P_{j^*}$ with {$e_{j^*}(c_r)\in P_{j^*}$}} given the choices of {$\bfm\pi_B$ by players in} $B$. That is, given $\bfm{\pi}_{B}$, player $j^*$  is not {weakly} preempted by  any player in {$\Delta(c_r)\sm(B\cup\{j^*\})$} {when he travels along $P_{j^*}$,}  regardless of  {the choices of players   in {$\Delta(c_r)\sm(B\cup\{j^*\})$}}. 
{Then combining with the minimalities of $B$ and $\tau^*$, we have $t_{j^*}^u(P_{j^*},\bfm{\pi}_{-j^*})=\tau^*$ and $j^*$ weakly preempts $i^*$ at vertex $u$ under path profile $(P_{j^*},\bfm{\pi}_{-j^*})$.}
Define an   $o_{j^*}(c_r)$-$d$ path {$R_{j^*}:=P_{j^*}\cup \Pi_{i^*}[u,d]\in\mathcal P[o_{j^*}(c_r),d]$}.
Then  
$t^d_{j^*}(R_{j^*},\bfm{\pi}_{-j^*})\le t^d_{i^*}(R_{j^*},\bfm{\pi}_{-j^*})$,  where the  equality can only happen when $u=d$. If $\exists v\in \Pi_{i^*}$ such that $t^v_{i^*}(R_{j^*},\bfm{\pi}_{-j^*})\neq t^v_{i^*}(\bfm{\pi})$, then by Lemma \ref{dominating} we know, $j^*$ dominates $i^*$ at vertex $v$ under $\bfm\pi$, which is a contradiction to Lemma \ref{nodominating} since $\bfm\pi$ is an NE. Otherwise $\forall v\in \Pi_{i^*}$, we have $t^v_{i^*}(R_{j^*},\bfm{\pi}_{-j^*})=t^v_{i^*}(\bfm{\pi})$. It follows that  $t^d_{j^*}(R_{j^*},\bfm{\pi}_{-j^*})\le t^d_{i^*}(R_{j^*},\bfm{\pi}_{-j^*})=t_{i^*}^d(\bfm{\pi})=\tau(\bfm\pi,k)<t_{j^*}^d(\bfm\pi)${, where} the last inequality follows from $j^*\not\in B${, contradicting} the fact that $\bfm\pi$ is an NE.  This proves
the correctness of (\ref{NEproperty:eq1}).

Once the players in {$B$} have chosen their strategies as specified by $\bfm{\pi}_{B}$, then by the correctness of (\ref{NEproperty:eq1}),  we can apply  {Algorithm \ref{alg:gernal}} {with $U:=B$ and $\bfm b:=\bfm\pi_B$},  {which provides us} a player in {$\Delta(c_r)\sm B$}, {denoted as $j'$}, who owns a {dominating}  {$o_{j'}(c_r)$-$d$ path in {$\mathcal P[o_{j'}(c_r),d]$}, denoted $R_{j'}$}, among players in {$\Delta(c_r)\sm B$}.  Therefore, by {Lemma \ref{lem:basedominant} and $j'\not\in B$}, for any {$j\in \Delta(c_r)\sm B$ and partial path profile $\bfm r$ of $\Delta(c_r)\sm B$}, we  have
\begin{equation*}
t_j^d(\bfm{\pi}_{B}, \bfm r)\ge t_{j'}^d(\bfm{\pi}_{B}, R_{j'})=t_{j'}^d(\bfm{\pi})=\tau(\bfm{\pi},k+1).\end{equation*}
{This implies that the third inequality of (\ref{NEproperty:eq1}) is valid. The fourth  {inequality} of (\ref{NEproperty:eq1}) is true due to {$i\in B$} and the first {equality} of (\ref{NEproperty:eq1}).}
\end{proof}

\begin{proof}[Proof of Theorem \ref{cor:nobetter'}]  (i) This is simply an interpretation of {$t^v_i(\bfm{\pi})=t^v_i(\bfm{\pi}_{B}, \bfm r)$ with $B=\Delta(\bfm\pi,[k])$ for each $k\ge1$}
in {Lemma \ref{NEproperty'}.}

(ii) For {each $k\ge1$, let $B:= \Delta(\bfm\pi,[k-1])$. The equality stated in (\ref{NEproperty:eq1}) enables us to apply Algorithm \ref{alg:gernal} and Lemma \ref{lem:basedominant}, which provides us a player $i\in \Delta(c_r)\sm B$} who possesses a dominating path $R_{i}$ provided the players in $ B$ follow their routes as in $\bfm\pi$. Let {$\tau^*$} denote the earliest time a player in $\Delta(c_r)\sm B$ reaches $d$ among all routings {of $\Gamma^n(c_r)$} in which players in {$ B$} take their routes as in $\bfm\pi$. It follows from {Lemma~\ref{lem:basedominant}} that  for any partial strategy profile~{$\bfm r$} 
{of} players in $\Delta(c_r)\sm B$ and any player $j\in\Delta(c_r)\sm B$ {it holds that}
\begin{center}
$t_j^d(\bfm{\pi}_{[k-1]}, \bfm r)\ge \tau^*=t_{i}^d(\bfm{\pi}_{[k-1]}, R_{i}, \bfm r_{\Delta(c_r)\,\sm\, (B\cup \{i\})})$.
\end{center}
In particular, {the arbitrary choice of $\bfm r$ gives}
\begin{center}
$t_j^d(\bfm{\pi})\ge \tau^*=t_{i}^d( R_{i}, \bfm\pi_{\Delta(c_r)
\setminus\{i\}})$.
\end{center}
Since $i$ {cannot} be better off via deviation to $R_{i}$, the minimality of $\tau^*$ enforces $t_{i}^d(\bfm\pi)=\tau^*$ and $i\in\Delta(\bfm\pi,k)$. 

 (iii) This can be shown easily by taking $\bfm r=\bfm{\pi}_{\Delta(c_r)\,\sm\,\Delta(\bfm\pi,[k])}$ in the {first inequality
 of (\ref{NEproperty:eq1}) in Lemma \ref{NEproperty'}.}

 (iv) Suppose on the contrary that there exists a set $S\subseteq\Delta(c_r)$ of players  who are able to strictly better off through deviating together from an NE $\bfm\pi$ of $\Gamma^n(c_r)$. Let $k$ be the smallest number such that $S\cap \Delta(\bfm{\pi},k)\neq \emptyset$. Due to 
 {Hierarchal Optimality in (ii),} all players in $\Delta(\bfm{\pi},k)$ obtain their optimal latencies provided that no player in $\Delta(\bfm{\pi},[k-1])$ deviates from $\bfm{\pi}$, a contradiction.
\end{proof}

Recalling original priorities defined in Definition \ref{def:original}, {let} players in $\Delta$ be indexed as $1,2,\ldots$ according to their original priorities (smaller indices correspond {to} higher priorities). We have the following straightforward corollary of {Lemma \ref{lem:dp} and} Global FIFO stated in {Theorem \ref{cor:nobetter'}}. 

\begin{corollary}\label{cor:NEproperty} If $\bfm\pi$ is an NE of $\Gamma^n$, then it satisfies the following properties:
\begin{mylabel6}
\item {{\em Weak Earliest Arrival.} For each player $i$ and a partial path profile of players other than $i$, player $i$ possesses a best response that is Earliest Arrival. This implies that any NE for a new game with the restriction that all players take Earliest Arrival paths is still an NE without this restriction.}

     \item {\em Consecutive Exiting.} The indices of players within the same batch under $\bfm \pi$ are consecutive. That is, if $i,j\in \Delta(\bfm \pi, k)$ with $i<j$, then $h\in \Delta(\bfm \pi, k)$ for all $i\le h\le j$.

    \item {\em Temporal Overtaking.} If   under $\bfm\pi$  player $i$ overtakes player $j$ at some vertex $v\in V\setminus \{o\}$ (i.e., $i$ enters $G$ with a lower original 
        priority but reaches $v$ earlier than $j$ does),  then under $\bfm\pi$ they reach the destination $d$ at the same time. Namely, the overtaking is temporal.  
\end{mylabel6}
\end{corollary}

\subsection{{Constructing {an} SPE from {a} given NE}}\label{apx:construct}

 {Recall that} $(e_i)_{i\in\Delta(c_{r-1})}$ is the action profile for players in $\Delta(c_{r-1})$ at configuration $c_{r-1}$ such that no {action in the profile deviates} from $\bfm\rho$. More specifically, for each $i\in\Delta(c_{r-1})$, letting $f_i$ denote the first edge of $P_i$, we have
 \begin{mylabel}
 \item[$\bullet$] $e_i=f_i$ if  $i$ queues on $e_i$ after someone else under $c_{r-1}$;
  \item[$\bullet$] $e_i$ ($\ne f_i$) is the  second  edge of $P_i$ if     $i$ queues first on $f_i$ under $c_{r-1}$;
   \item[$\bullet$] $e_i$ is a null action if   $i$ queues first on $f_i$ under $c_{r-1}$ and $f_i$'s head is  $d$.\end{mylabel} 

{For} every $i \in\Delta(c_r)$, we observe that either $e'_i=f_i$ and $o_i(c_r)=o_i(c_{r-1})$, or $e'_i\ne f_i$ and $o_i(c_r)$ is the head of $f_i$. Therefore, given any $o_i(c_r)$-$d$ path $X_i\in\mathcal P[o_i(c_r),d]$, $\{f_i\}\cup X_i$ is an $o(c_{r-1})$-$d$ path in $\mathcal P[o_i(c_{r-1}),d]$.

 Recall {Construction  I with} the definition of $\Bbbk$ in {(\ref{eq:defk}), Construction II} and all the other settings in Subsection \ref{sec:relation}. Observe that either $\Delta(c_{r-1})=\Delta(c_r)$, or $\Delta(c_{r-1})\sm\Delta(c_r)\ne\emptyset$  and each player in $\Delta(c_{r-1})\sm\Delta(c_r)$ exits $\bar G$ at time $r$, giving  $\Delta(c_{r-1})\sm\Delta(c_r)=\Delta(\bfm\rho,1)\subseteq\Delta(\bfm\rho,[\Bbbk])$. In any case we have
 \begin{equation}\label{eq:subset}
\Delta(c_{r-1})\sm\Delta(c_r)\subseteq\mathbb B:=\Delta(\bfm\rho,[\Bbbk]).
 \end{equation}

 Moreover, we recall that $\bfm\rho$ is an NE of $\Gamma^n(c_{r-1})$ and the Hierarchal Independence of $\bfm\rho$ stated in {Theorem \ref{cor:nobetter'}:} 
 as long as the chosen paths of players in $\mathbb B$ remain as in   $\bfm\rho_{\mathbb B}$, no matter what paths the players in $\Delta(c_{r-1})\sm\mathbb B$ choose, they have no impact on the arrival time of any player in $\mathbb B$ at any vertex. This implies the following important property, which is the base of Construction II.

 \begin{lemma}\label{lem:bbbk}
 For any  partial  path profile  $\bfm\chi=(X_j)_{i\in\Delta(c_r)\setminus \mathbb B}$, where $X_j$ can be any  $o_i(c_{r})$-$d$ path in $\mathcal P[o_j(c_{r}),d]$, any player $i\in \mathbb B\cap\Delta(c_{r})$ and any vertex $v\in P_i'$, it holds that
 $t_i^v(\bfm\rho)=t_i^v(\bfm\rho_{ \mathbb B},(\{f_j\}\cup X_j)_{j\in\Delta(c_r)\setminus\mathbb B})=t_i^v((P'_j)_{j\in\mathbb B\cap\Delta(c_r)},\bfm\chi)$.
   \end{lemma}

  \begin{proof}Notice from (\ref{eq:subset}) that $(\bfm\rho_{ \mathbb B},(\{f_j\}\cup X_j)_{j\in\Delta(c_r)\setminus\mathbb B})$ is a strategy profile 
{of} game {$\Gamma^n(c_{r-1})$}. So the first equality of the conclusion follows from {the Hierarchal Independence in Theorem \ref{cor:nobetter'}(i).} 
The second equality is straightforward from the fact that each player in $\Delta(c_{r-1})\sm\Delta(c_r)$ (if any) belongs to $\Delta(\bfm\rho,1)\subseteq \mathbb B$, and he {has} only a null action under $c_{r-1}$ in any case, which has no {effect on} other players.
  \end{proof}

  For easy expression of null actions $e'_j$ of players in $j\in \Delta(c_{r-1})\sm\Delta(c_r)$, we reserve symbol $\bfm\phi$ for the profile $(e'_j)_{j\in\Delta(c_{r-1})\setminus\Delta(c_r)}$ of null actions.

\begin{lemma}\label{lem:inductive}
$\bfm\rho'$ is an NE of game {$\Gamma^n(c_{r})$}.
\end{lemma}
\begin{proof}
We need to prove that {$t^d_i(\bfm\rho')\le t^d_i(X_i', \bfm\rho_{\Delta(c_r)\,\sm\,\{i\}}')$} holds for all player $i\in\Delta(c_r)$ and all $o_i(c_r)$-$d$ path $X'_i\in\mathcal P[o_i(c_r),d]$.

{\bf Case 1.} $i\in  \mathbb B$. Suppose $i\in\Delta(\bfm\rho,k)$ for some $k\le \Bbbk$. Therefore, for any path profile $\bfm\chi=(X_j)_{j\in\Delta(c_r)}$ of $\Gamma^n(c_r)$, {with $B:=\Delta(\bfm\rho,[k-1])$} we have
\begin{eqnarray*}
t_i^d(\bfm\rho') 
      &=& t^d_i(\bfm\rho)\\
      &\le& t^d_i(\bfm\rho_{ B},  (\{f_j\}\cup X_j)_{j\in\Delta(c_{r})\setminus  B},\bfm\phi_{\Delta(c_{r-1})\setminus \Delta(c_r)\setminus B})\\
&=&t^d_i(\bfm\rho'_{ B\cap\Delta(c_r)},\bfm\chi_{\Delta(c_r)\setminus B}),
\end{eqnarray*}
where  the first equality is
by Lemma \ref{lem:bbbk},
 the inequality is from  Hierarchal Optimality in {Theorem \ref{cor:nobetter'}(ii)} 
 and the last equality is due to again the definition  of $\bfm\rho'$. In particular, when taking $X_i=X'_i$ and $X_j=P'_j$ for every $j\in\Delta(c_r)\sm \{i\}$, we obtain $t^d_i(\bfm\rho')\le t^d_i(X_i', \bfm\rho_{\Delta(c_r)\,\sm\,\{i\}}')$.

 {\bf Case 2.}  $i\in \Delta(c_r)\sm \mathbb B$. It can be seen from {Lemma \ref{lem:basedominant}} that the path $P_j$ of each player $j\in\Delta(c_r)\sm \mathbb B$  is his  best response to other players' choices, giving $t^d_i(\bfm\rho')\le t^d_i(X'_i,\bfm\rho_{\Delta(c_r)\,\sm\,\{i\}}')$. 
\end{proof}


\begin{proof}[{Proof of Theorem \ref{thm:equilibrium-relation}}]  Let $\sigma=(\sigma_i)_{i\in\Delta}$ be 
{a strategy profile of} 
$\Gamma$ defined as follows: at each {history $h_r=(c_0,\ldots,c_r)$}, players in $\Delta(c_r)$ take actions as specified by the NE  {$\bfm\pi(h_r)$  constructed in Section \ref{sec:relation} for $h_r$}, where $\bfm\pi(c_0)=\bfm\pi$. Similar to the proof of Theorem \ref{thm:SPE}, it can be deduced from Construction I (i.e., (\ref{eq:construct1})) and Construction II {(and Lemma \ref{lem:basedominant})} that for each {history $h_r$} the path profile induced by {$h_r$} and $\sigma$ is exactly {$\bfm\pi(h_r)$}.

To see that $\sigma$ is an SPE of $\Gamma$, we fix an arbitrary $r\ge0$ and an arbitrary history $h_r=(c_0,\ldots,c_r)\in\mathcal H_r$. Let $\bfm\rho'=(P'_i)_{i\in\Delta(c_r)}$ denote the NE {$\bfm\pi(h_r)$ of $\Gamma^n(c_r)$} we {have constructed} for {$h_r$}. In case of $r=0$, we set $\bfm\rho':=\bfm\pi$. Moreover, we consider any $i\in\Delta(c_r)$, any $\sigma'_i\in\Sigma_i$, and the path profile {$\bfm\chi=(X_j)_{j\in\Delta(c_r)}$} induced by {$h_r$} and $\sigma':=(\sigma'_i,\sigma_{-i})$. We need {to verify that}
$t_i(\sigma|{h_r})\le t_i(\sigma'|h_r)$.

If $r=0$, then {we} suppose that $i\in\Delta(\bfm\pi(c_r),k)$ {and write $B=\Delta(\bfm\pi,[k-1])$}. Construction~I, i.e., (\ref{eq:construct1}), implies that {$X_j=P'_j$ for all $j\in B$}, and in turn Hierarchal Optimality in 
{Theorem \ref{cor:nobetter'}(ii)} says that $t_i(\sigma|{h_0})=t^d_i(\bfm\pi)\le t^d_i(\bfm\pi_{ B},\bfm\chi_{\Delta\setminus B})=t^d_i(\bfm\chi)=t_i(\sigma'|h_0)$.

So we assume {now} $r\ge1$. Therefore {$h_r$ is a child history of some (unique) history $h_{r-1}=(c_0,\ldots,c_{r-1})\in\mathcal H_{r-1}$}. Let {$\bfm\rho=(P_j)_{j\in\Delta(c_{r-1})}$} denote the {NE $\bfm\pi(h_{r-1})$ of $\Gamma^n(c_{r-1})$},  $\Bbbk$  be defined as in (\ref{eq:defk}),   {and $\mathbb B:=\Delta(\bfm\rho,[\Bbbk])$}.

If $i\in\Delta(\bfm\rho,k)$ for some $k\le \Bbbk$, then (\ref{eq:construct1})  implies that {$X_j=P_i[o_i(c_r),d]=P'_j$} for all $j\in B\cap\Delta(c_r) $, {where $B=\Delta(\bfm\rho,[k-1])$.  As} in Case 1 of the proof of Lemma~\ref{lem:inductive} we deduce that $t_i(\sigma|{h_r})=t^d_i(\bfm\rho')\le t^d_i(\bfm\rho'_{ B\cap\Delta(c_r)},\bfm\chi_{\Delta(c_r)\setminus B})=t^d_i(\bfm\chi)=t_i(\sigma'|h_r)$.

It remains to consider the case of $i\in \Delta(c_r)\sm\mathbb B$. Assume  {w.l.o.g.} that $i$ is exactly the $i$th  player in the ordering {$1,2,\ldots $} of players in $\Delta(c_r)\sm\mathbb B$ associated with the iterative dominating path profile constructed in Construction II. Again (\ref{eq:construct1}) guarantees $\bfm\chi_{\mathbb B\cap\Delta(c_r)}=\bfm\rho'_{\mathbb B\cap\Delta(c_r)}$. It follows from Lemma \ref{lem:basedominant} that $\bfm\chi_{[i-1]}=\bfm\rho'_{[i-1]}$,  and $t_i(\sigma|{h_r})=t^d_i(\bfm\rho')\le t^d_i(\bfm\rho'_{\mathbb B\cap\Delta(c_r)},\bfm\rho'_{[i-1]},\bfm\chi_{\Delta(c_r)\setminus\mathbb B\setminus [i-1]})=t^d_i(\bfm\chi)=t_i(\sigma'|h_r)$, which completes the proof.
\end{proof}

\section{Details in Section 5}

\begin{proof} [{Proof of Lemma \ref{lem:degree}}]
By symmetry, it suffices to prove $n_1 \le2\Lambda L(2\Lambda+n_2)$. For $i=1,2$, let $S_i$ denote the set of
$n_i$ players who are inside $G_i=(V_i,E_i)$ at time $t$.  {Denote by $o'$ and $d'$ the origin and destination nodes of
the resulting network from connecting $G_1$ and $G_2$ in parallel.}

Take $\zeta $ to be a fixed player in $S_1$ who enters $G_1$  {the} latest. Then all the players in
$S_1\sm\Delta^{o'}_\zeta$ reach $o'$  before time $t^{o'}_\zeta$. Among them, we choose a subset $H$ of $\Lambda$ players     who reach $d'$ under $\bfm \pi$ as late as possible ({note that} the players in $H$ {cannot} reach $d'$ {at the same time}).  It follows that each $h\in H$ reaches $d'$ later than {those in} 
$S_1\sm(\Delta^{o'}\cup\{\zeta\}\cup\Delta^{d'}_h\cup H)$. Since  each time at most $|E_1^-(d')|$ players in $G_1$ can reach $d'$ {at the same time, and} $|E_1^-(d')|$, the {in-degree} of $d'$ in $G_1$, is at most $\Lambda-1$, we have
\begin{equation}\label{eq:g1}
t^{d'}_h(\bfm\pi)\ge t+\left\lceil\frac{|S_1\sm(\Delta^{o'}_\zeta\cup\Delta^{d'}_h\cup H|)}{|E_1^-(d')|}\right\rceil
\ge t+\frac{n_1-3\Lambda}{\Lambda-1}, \ h\in H.
\end{equation}
Moreover, from the choice of $\zeta$ and $H\subseteq S_1\sm\Delta^{o'}_\zeta$ we derive
\begin{equation}\label{eq:eta}
t^{o'}_h(\bfm\pi)<t^{o'}_\zeta(\bfm\pi)\le t, \  h\in H.
\end{equation}

Since $h$ reaches $o'$ earlier than $\zeta$, it is  {immediate} from the {FIFO} property
in {Theorem \ref{cor:nobetter'}(iii)} 
that $t^{d}_h(\bfm\pi)\le t^{d}_\zeta(\bfm\pi)$ for all $h\in H$. Since
$|\{\zeta\}\cup H|=|\Lambda|+1>|\Delta^d_\zeta|$, there exists $\eta\in H$ such that
\begin{equation}\label{eq:strict}
 t^{d}_\eta(\bfm\pi)< t^{d}_\zeta(\bfm\pi).
\end{equation}

Next, we distinguish between two cases depending on whether there exists a player who enters $G_2$
{within time period} $(t^{o'}_\zeta(\bfm\pi),t+L(2\Lambda+n_2)]$.

\paragraph{Case 1.} Under $\bfm \pi$ no player enters $G_2$ {within time period} $(t^{o'}_\zeta(\bfm\pi),
t+L(2\Lambda+n_2)]$. {Consider} the routing profile $\mathbf r=(R_\zeta,\bfm\pi_{\Delta\sm\{\zeta\}})$,
in which $\zeta$ changes the part of his route  from $\Pi_\zeta[o',d']\subseteq G_1$ to an $o'$-$d'$ path
$R_\zeta[o',d']\subseteq G_2$, while the other parts of his route, $R_\zeta[o,o']=\Pi_\zeta[o,o']$ and $R_\zeta[d',d]
=\Pi_\zeta[d',d]$, are kept unchanged. By (\ref{eq:strict}) and  {Lemma \ref{NEproperty'},} 
we see that $\zeta$
cannot preempt $\eta$  ({i.e.,} reach some vertex earlier than $\eta$) under $\mathbf r$, which along
with (\ref{eq:g1}) gives $t+\frac{n_1-3\Lambda}{\Lambda-1}\le t^{d'}_\eta(\bfm\pi)\le t^{d'}_{\zeta}(\mathbf r)$. If
$t^{d'}_{\zeta}(\mathbf r)\le t+L(2\Lambda+n_2)$, then $n_1\le 3\Lambda+(\Lambda-1)L(2\Lambda+n_2)$ and we are done.
So we assume
\begin{equation}\label{eq:ass}
t^{d'}_{\zeta}(\mathbf r)> t+L(2\Lambda+n_2).
\end{equation}

Let $S$ consist of players who enter $G_2$  {no later than} $t^{o'}_\zeta(\bfm\pi)$ under $\bfm\pi$. {According
to the condition} of Case~1, inequality (\ref{eq:ass}) and the series-parallel structure of $G$, it is not hard to see that
from time $t^{o'}_\zeta(\mathbf r)=t^{o'}_\zeta(\bfm\pi)$ to {time} $t+L(2\Lambda+n_2)$, no player outside $S$
can queue before $\zeta$ {at} the same edge at the same time of $\zeta$ under $\mathbf r$.
Since all players in $S\sm\Delta^{o'}_\zeta$ preempt $\zeta$ under $\bfm\pi$, it follows from the FIFO property in {Theorem \ref{cor:nobetter'}(iii)} 
that the players of $S\sm(\Delta^{o'}_\zeta\cup\Delta^d_\zeta)$ all reach $d$  earlier than time $t^d_\zeta(\bfm\pi)$.
Then by 
{Lemma \ref{NEproperty'},} players in $\Delta^d_\zeta$  {cannot} preempt any player of
$S\sm(\Delta^{o'}_\zeta\cup\Delta^d_\zeta)$ under $\mathbf r$. Since at time $t$  the players 
{in} $S\sm(\Delta^{o'}_\zeta
\cup\Delta^d_\zeta)$ who are inside $G_2$ is a subset of $S_2$,  we see that {under $\mathbf r$ player}
$\zeta$ reaches $d'$ {no later than} time
$t+|E(R_\zeta[o',d'])|\cdot|S_2\cup\Delta^{o'}_\zeta\cup\Delta^{d}_\zeta|\le t+L(n_2+2\Lambda)$, 
{because}
{under $\mathbf r$ by the time $ t+L(n_2+2\Lambda)$ the transit time of $\zeta$}
through each edge of $E(R_\zeta[o',d'])$ is at most $|S_2\cup\Delta^{o'}_\zeta\cup\Delta^{d}_\zeta|$,
a contradiction to (\ref{eq:ass}).

\paragraph{Case 2.} There exists $\psi\in \Delta$ who enters $G_2$ at an earliest time
$t^{o'}_\psi(\bfm\pi)\in(t^{o'}_\zeta(\bfm\pi),t+L(2\Lambda+n_2)]$. If  {a player} $j\in \Delta$ preempts $\psi$
at some time point in $[t^{o'}_\psi(\bfm\pi),t^{d'}_\psi(\bfm\pi)]$ under $\bfm\pi$, then, {by Theorem \ref{cor:nobetter'}(iii),} 
either $j\in\Delta^d_{\psi}$ or $j$ enters $G_2$ no later than $\psi$ does. The definition of $\psi$ (i.e., the minimality of
$t^{o'}_\psi(\bfm\pi)$) implies that either $j\in\Delta^d_\psi\cup\Delta^{o'}_\psi$ or $j$ enters $G_2$ at time
$t^{o'}_j(\bfm\pi)\le t^{o'}_\zeta(\bfm\pi)\le t$. So either $j\in\Delta^d_\psi\cup\Delta^{o'}_\psi$, or $j\in S_2$,
or $t^{d'}_j\le t$ {($j$ leaves $G_2$ at time $t$)}. It follows that
\begin{equation}\label{eq:psi}
t^{d'}_\psi(\bfm\pi)\le \max\{t,t^{o'}_\psi(\bfm\pi)\}+|E(\Pi_\psi[o',d'])|
\cdot|\Delta^d_\psi\cup\Delta^{o'}_\psi\cup S_2|\le t+2L(2\Lambda+n_2).
\end{equation}
Recalling  {the definition of $\psi$ and} (\ref{eq:eta}), we have $t^{o'}_\psi(\bfm\pi)>t^{o'}_\zeta(\bfm\pi)>
t^{o'}_h(\bfm\pi)$ for all $h\in H$. If $t^{d'}_\psi(\bfm\pi)<t^{o'}_h(\bfm\pi)$ for all $h\in H$, then we deduce from
 {Consecutive Exiting in Corollary \ref{cor:NEproperty}(ii)} that all players in $\{\psi\}\cup H$ reach $d$ at the same time under $\bfm\pi$, i.e.,
$t^d_\psi(\bfm\pi)$. However, $\{\psi\}\cup H\subseteq\Delta^d_{\psi}$ and $|\{\psi\}\cup H|=\Lambda+1$ contradict
$|\Delta^d_\psi|\le\Lambda$. {Hence,} we have $t^{d'}_\psi(\bfm\pi)\ge t^{d'}_h(\bfm\pi)$ for some $h\in H$. {In turn,}
(\ref{eq:psi}) and (\ref{eq:g1}) give $ t+2L(2\Lambda+n_2)\ge t^{d'}_\psi(\bfm\pi)\ge t+\frac{n_1-3\Lambda}{\Lambda-1}$,
yielding $n_1\le3\Lambda+2L(\Lambda-1)(2\Lambda+n_2)\le2\Lambda L(2\Lambda+n_2)$, as desired.
\end{proof}


\begin{proof} [{Proof of Lemma \ref{lem:series1}}]
Suppose $G_s$ is obtained from $G_1$ and $G_2$ by identifying $d_1$ and $o_2$.
We distinguish between two cases.

In {the} case of $|\Xi(G_1)|\le |\Xi(G_2)|$, clearly $\Xi(G_s)$ is just $\Xi(G_1)$. To prove $G_s\in\mathfrak F$, as $G_1,G_2\in\mathfrak F$, it suffices to show that  {the number of players $G_2^l$ can accommodate is upper bounded} at any time. Since $G_1^r$ can accommodate at most $F^r_1$ players, and $\Xi(G_2)$ is full as long as $G_2^l$ accommodates more than $F^l_2$ players, we deduce that $\Xi(G_2)$ is full as long as $G_1^r\cup G_2^l$ accommodates more than $F^r_1+F^l_2$ players. Consider  any time point $t$ such that $G_1^r\cup G_2^l$ accommodates at most $F^r_1+F^l_2$ players at time $t$, and more than $F^r_1+F^l_2$ players at time $t+1$. {Then,}  {at time $t+1$, we know} $G_1^r\cup G_2^l$ accommodates at most $F^r_1+F^l_2+|\Xi(G_1)|$ players and $\Xi(G_2)$ is full. So $|\Xi(G_2)|$ players   leave $G_1^r\cup G_2^l$ at time $t+2$ while at most $|\Xi(G_1)|$ players enter $G_1^r\cup G_2^l$. It follows from $|\Xi(G_1)|\le|\Xi(G_2)|$ that the number of players inside $G_1^r\cup G_2^l$ is nonincreasing unless the number decreases below $F^r_1+F^l_2+1$. {Therefore, at any time, $G^r_1\cup G^l_2$ and  {consequently} $G^l_2$} can accommodate at most $F^r_1+F^l_2+|\Xi(G_1)|$ players.

In {the} case of $|\Xi(G_1)|> |\Xi(G_2)|$, we see that $\Xi(G_s)$ is just $\Xi(G_2)$. To prove $G_s\in\mathfrak F$, we only need to show that $\Xi(G_2)$ is full as long as  $G_s^l$ ($=G_1\cup G_2^l$) accommodates more than a certain finite number  of players.  {Let $t$ be a time when $\Xi(G_2)$ is not full and $G_1\cup G_2^l$ accommodates $F$ ($>F^l_1+F^r_1+F_2^l$) players, then it must be the case that $F\le F^l_1+(1+\Lambda)( F_1^r+F_2^l)$. Since $G_2\in\mathfrak F$, we know $G_2^l$ accommodates 
 {fewer} than $F^l_2$ players at time $t$. It follows that $G_1^l$ accommodates at least $F-(F^r_1+F_2^l)>F^l_1$ players.   Let $t'$ be the latest time before $t$ when $G_1^l$ accommodates at most $F^l_1$ players. Since from $t'+1$ to $t$,  at most $\Lambda$ players can enter $G_s$ from $o_1$ at each time, then}
\[
F-(F^r_1+F_2^l)-F^l_1\le (t-t')\Lambda.
\]
Since, by the choice of $t'$, at time $t'+1,t'+2,\ldots,t$, there are more than $F^l_1$ players in $G^l_1$, cut $\Xi(G_1)$ is full at all these consecutive $t-t'$ time points. Therefore, at time $t$, the number of players inside $G_1^r\cup G_2^l$ is at least
\[
(t-t')(|\Xi(G_1)|-|\Xi(G_2)|)\ge t-t'\ge (F-F^l_1-F^r_1-F_2^l)/\Lambda.
\]
On the other hand, by $G_1\in\mathfrak F$ and the choice of $t$, there are at most $F_1^r+F_2^l$ players inside $G_1^r\cup G_2^l$ at time $t$. 
{It follows from $(F-F^l_1-F^r_1-F_2^l)/\Lambda\le F_1^r+F_2^l$ that $F\le F^l_1+(1+\Lambda)( F_1^r+F_2^l)$, proving the lemma}.
\end{proof}


\begin{proof} [{Proof of Theorem \ref{thm:queue-length-SPE}}]
{Given the conditions in the theorem, it is easy to see that $|\bar{E}^+(v)|\ge |\bar{E}^-(v)|$
for any vertex $v\in V\sm\{o,d\}$.} Since $G$ is acyclic, there is a natural {full} order {$\prec$ on its} vertices.
For all {$u,v\in V$}, if there is a {$u$-$v$} path in $G$, we  say that $v$ is larger than $u$. We may also say that
an edge {$e\in E$} is larger than $u$ if there is a directed  path {in $G$} starting from $u$ and {going} through
$e$.

Note that the theorem is equivalent to stating that there is an {upper bound on} the lengths of the queues on all edges {of $G$. For easy reference, we} call an edge that has no such {upper bound} an {\em infinity edge}.
  Using the partial order $\prec$, we prove by induction to show that each edge has an {upper bound}. The base cases are easy. In fact, for each maximal vertex {$u\in V\sm\{d\}$}, {by the apparent boundedness on parallel networks,}  no {outgoing edge from} $u$ is an infinity edge.

For any {$u\in V\sm\{d\}$}, suppose now there is no {vertex of $G$} that is larger than $u$ and has an infinity {outgoing edge. Thus, for each edge {$e\in E\sm E^+(u)$}} that is larger than $u$, there exists an {upper bound, denoted $U_e$,  on} the queue lengths on $e$. For each out-neighbor $v$ of $u$, {we write $\mathscr P[v,d]$ for  the set of  $v$-$d$ paths in $G$, and  observe that {$U_v:=\max_{P\in \mathscr P[v,d]}\sum_{e\in P}(U_e+1)$} is an upper bound  on the costs} of  all  the {$v$-$d$ paths}.

Let {$U_u:=1+\max_{v } U_v$, where the maximum is taken over the set of all out-neighbors $v$ of $u$.} Since the traffic is {induced by an SPE},  at any time, the {difference} between any two queue lengths of the {outgoing} edges of $u$ is at most $U_u$. Therefore, whenever the total number of players on the {outgoing} edges of $u$ exceeds $(|E^+(u)|-1)U_u$, each {outgoing edge of $u$ has a nonempty queue on it}. Since the {out-degree} of $u$ is no less than {its in-degree in $\bar G$}, it can be seen that the total number of players on the {outgoing} edges of $u$ does not increase in the next time step, and hence it never exceeds  {$(|E^+(u)|-1)U_u+|E^-(u)|$}. This gives an {upper bound} for the queue lengths of the {outgoing} edges of $u$ {and} finishes the proof.
\end{proof}

\end{appendix}

\end{document}